\documentclass[english]{cccconf}
\usepackage{color}
\usepackage{cite}
\usepackage{amsmath,amssymb,amsfonts}
\usepackage{diagbox}
\usepackage{textcomp}
\usepackage{tikz}
\graphicspath{{images/}}
\usepackage{graphicx}
\usepackage{epstopdf}
\usepackage{amssymb}
\usepackage{algorithm}
\usepackage{algpseudocode}
\usepackage{amsmath}
\allowdisplaybreaks

\makeatletter

\newcommand{\Rmnum}[1]{\expandafter\@slowromancap\romannumeral #1@}

\def\0{{\bf 0}}

\def\thefootnote{*}

\def\proof{\noindent{\bf Proof: }}
\newtheorem{theorem}{Theorem}[section]
\newtheorem{lem}[theorem]{Lemma}

\newtheorem{definition}{Definition}[section]
\newtheorem{prp}[theorem]{Proposition}

\newtheorem{rem}[theorem]{Remark}
\begin{document}
\title{Aggregation of Evolutionary Game Dynamics on Large-Scale Weighted Networks}
\author{Zhiru Wang, Bin Wu$^\ast$
}
\maketitle

\begin{abstract}
Networked evolutionary games, which integrate network topology and game dynamics, serve as a powerful framework for complex systems. Evolutionary games on large-scale networks, however, have been analytically challenging due to the curse of dimensionality arising from large network size. This paper proposes an aggregation method based on backward equivalence, such that agents within the same equivalence class behave exactly the same over time.
We give a necessary and sufficient condition under which the weighted networked evolutionary games (WNEG) are reduced to an equivalent system with low dimension. The aggregation is shown to reduce the computational burden ranging from strategy consensus, strategy optimization, controllability, to optimal control of the WNEG. Examples are  provided. Our work opens an avenue to solve the curse of dimensionality on networked evolutionary game systems with mathematical rigor.
\end{abstract}

\footnotetext{$\dag$: Corresponding author.
Zhiru Wang and Bin Wu are with Key Laboratory of Mathematics and Information Networks, School of Mathematical Sciences, Beijing University of Posts and Telecommunications, 100876, Beijing, P.R. China (e-mail: wzrsdmn2023@163.com; bin.wu@bupt.edu.cn).
}

\section{Introduction}
\label{sec:introduction}
Evolutionary games with spatial structure or spatial games are pioneered by Nowak and May \cite{M1}. It is shown that spatial structure facilitates the emergence and maintenance of cooperative behavior for the Prisoner's Dilemma on two-dimensional square lattice. Networked evolutionary games \cite{G2} have been attracting increasing attention in control community \cite{J3,D4,B5}. Heterogeneity of individual interactions is ubiquitous. The distance between sensor and sink in a sensor network is crucial for data transmission \cite{T6}. If the wireless communication link is weighted by the distance between nodes,
the sensor network is a weighted network. Similarly, the travel time is imperative for optimal travel routes in a transportation network  \cite{G39}. If each road is weighted by its travel time, the transportation network is a weighted network too.
There has been increasing interests in weighted networked evolutionary game (WNEG) with its wide application in engineering systems
\cite{I40,J41}.

A WNEG is a finite-valued system thanks to its finiteness: (1) the number of players is finite; (2)the numbers of strategies are finite. Any finite-valued system can be transformed to a discrete-time linear system \cite{L9,C10,J11,L12} via semi-tensor product (STP) \cite{D8}.
The evolutionary dynamics of
the WNEG is driven by a logical matrix, i.e., the profile transition matrix. The STP approach is a rigorous mathematical framework to study the dynamics and control WNEGs
\cite{J7,R13,Y14,R15}.
However, STP requires dimension expansion via tensor product for matrix multiplication. It leads to a sharp expansion in matrix size as the node size of network increases. This expansion is likely to exceed the computation capacity for the STP-based criterion.
For strategy optimization, consensus, stability and stabilization, the proposed criterion typically requires multiple matrix product operations on the profile transition matrix \cite{P44,Y45,J46}. When the number of players is $n$ and each player has $k$ strategies,
the computational complexity required to validate this type of criterion is $O(k^{3n})$.
It grows exponentially with the number of players.
Consequently, this type of criterion cannot handle WNEGs with more than thirty players in a reasonable amount of time even for two-strategy games \cite{K16}. This is the curse of dimensionality arising from STP.
It is thus timely to figure out a method to simplify the network topology before representing the dynamics of WNEGs as linear systems, if STP is adopted.

This paper proposes an aggregation method based on backward equivalence for WNEGs under the myopic best response adjustment rule, which aims to solve the curse of dimensionality by STP. The contributions of this paper are three-fold:
\begin{enumerate}
\item[$\mathrm{(i)}$] A necessary and sufficient condition is given to show under what edge connection patterns the network of the WNEG can be aggregated in the sense of backward equivalence.
\item[$\mathrm{(ii)}$]  The criteria are provided to check whether the WNEG with \emph{external control} can be reduced to a \emph{controlled} system with low dimension.
\item[$\mathrm{(iii)}$] The effectiveness of the proposed  aggregation for dimension reduction is quantified by the time complexity and space complexity ranging from strategy consensus, strategy optimization, controllability to optimal control of the WNEG.
\end{enumerate}
The rest of this article is organized as follows. Section 2 introduces notions and preliminaries on STP. The problem formulation is in Section 3. Section 4 is on how to make use of backward equivalence to reduce dimension for the WNEG. The effectiveness of the proposed aggregation method is demonstrated in Section 5. Finally, Section 6 gives a conclusion.
\section{Preliminaries}
Some necessary notations are listed below.
\begin{itemize}
\item $\mathbb{R}:$ the set of real numbers. $\mathbb{R}_{p\times q}:$ the set of $p\times q$ real matrices.
\item $M^T$ : the transpose of matrix $M$.
  \item Given a matrix $M\in\mathbb{R}_{p\times q}$, denote the $(i,j)$-th element and the $i$-th column (row) of $M$ by $[M]_{ij}$ and $\mathrm{Col}_i(M)\ (\mathrm{Row}_i(M))$, respectively.
\item $\mathbf{0}_{n}:=[\underbrace{0 ,\cdots,0}_{n}]^{T}$ and $\mathbf{1}_{n}:=[\underbrace{1 ,\cdots,1}_{n}]^{T}$.
\item $\mathcal{D}_{k}:=\{1,2,\cdots,k\}$ and $\mathcal{D}_{k}^{n}:=\underbrace{\mathcal{D}_{k}\times\cdots\times\mathcal{D}_{k}}_{n}$.
  \item $\Delta_n:=\{\delta_n^j|j=1,2,\cdots,n\}$, where $\delta_{n}^{j}=\mathrm{Col}_{j}(I_{n})$ and $I_n$ denotes the $n$-dimensional identity matrix.
  \item $\delta_{n}\{j_{1},j_{2},\cdots,j_{n}\}:=\{\delta_{n}^{j_{1}},\delta_{n}^{j_{2}},\cdots,\delta_{n}^{j_{n}}\}$.
  \item A matrix $M\in\mathbb{R}_{p\times q}$ is
  called a logical matrix, if $M=[\delta_{p}^{j_{1}},\delta_{p}^{j_{2}},\cdots,\delta_{p}^{j_{q}}]$; We can express the matrix $M$ as
  $M=\delta_{p}[{j}_{1},j_{2},\cdots,j_{q}]$,
  and define the set of all
  $p\times q$ logical matrices as ${\mathcal{L}}_{p\times q}$.
\item A $p\times q$ matrix $M$ is called a Boolean matrix, if $[M]_{ij}$ = 0 or 1 for all $i=1,2,\cdots,p,\;j=1,2,\cdots,q$, and denote
the set of all $p\times q$ Boolean matrices as $\mathfrak{B}_{p\times q}$.
\item $M,N\in\mathbb{R}_{p\times q}, M\geq N$ means that $[M]_{ij}\geq[N]_{ij}, \forall 1\leq i\leq p$ and $1\leq j\leq q$.
 \item $[a:b]$: the set of integers $\rho$ satisfying $a\leq\rho\leq b$.
 \item $|\Phi|$: the cardinality of set $\Phi$.
\end{itemize}

Next, we review the fundamentals of STP.
\begin{definition} (\hskip -0.1mm\cite{D8})\label{d2.2.1} The STP of two matrices $M\in {\mathbb{R}}_{p\times q}$ and $N\in {\mathbb{R}}_{m\times n}$ is defined as $$M\ltimes N=(M\otimes I_{\frac{\alpha}{q}})(N\otimes I_{\frac{\alpha}{m}}),$$ where $\alpha=lcm(q,m)$ is the least common multiple of $q$ and $m$, and $\otimes$ is the Kronecker product.\end{definition}

Evidently, when $q=m$, STP of matrices degenerates into the conventional matrix product. In what follows, we omit the symbol ``$\ltimes$" if no confusion arises.


Identify $i\thicksim\delta^{i}_{k},i=1,2,\cdots,k$, then $\mathcal{D}_{k}\thicksim\Delta_{k}$.
The following lemma is obtained.
{\lem (\hskip -0.1mm\cite{D8})\label{2.2.3} Let $g:\mathcal{D}^{n}_{k}\rightarrow \mathbb{R}$ ~(or $g:\mathcal{D}_{k}^n\rightarrow \mathcal{D}_{m}$) be a finite-valued logical function.  Then there exists a unique matrix $G\in\mathbb{R}_{1\times k^n}$ (or $G\in\mathcal{L}_{m\times k^n}$), such that
$$
g(x_1,x_2,\cdots,x_n)=G\ltimes_{i=1}^nx_i,
$$
where $G$ is called the structural matrix of $g$.
}

\section{Problem formulation}
\begin{definition}
A weighted networked evolutionary game (WNEG) is denoted by $((N,A),G,f)$, consists of
\begin{enumerate}
\item[$\mathrm{(i)}$]a weighted network $(N,A)$. $N=\{1,2,\cdots,n\}$ is the set of players.  $A:=(a_{ij})_{n\times n}$ is the adjacency matrix, where $a_{ij}=a_{ji}\geq0$. $a_{ij}>0$ implies that there is a edge between $i$ and $j$ with weight $a_{ij}$.  $a_{ij}=0$ means that there is no edge between $i$ and $j$. Denote $\mathcal{N}_i=\{j|a_{ji}>0\}$ as the set of neighbors of player $i$.
\item[$\mathrm{(ii)}$]a fundamental networked game (FNG) $G$. If $a_{ij}>0$ then players $i$ and $j$ play the FNG with strategy set $S_i=S_j=S_0=\{1,2,\cdots,k\}$. Let $C:=(c_{pq})_{k\times k}$ be the payoff matrix of the FNG, where $c_{pq}$ denotes the payoff of a player when it chooses strategy $p$ and its opponent chooses strategy $q$.
\item[$\mathrm{(iii)}$] the myopic best response adjustment rule $f:=\{f_1,f_2,\cdots,f_n\}$. The evolutionary equation for each player is
\begin{equation}\label{1.1}x_i(t+1)=f_i(x_j(t)|j\in\mathcal{N}_i),\end{equation} where $x_i(t), i\in N$ denotes the strategy of player $i$ at time $t$, $f_i: \mathcal{D}^{|\mathcal{N}_i|}_k\rightarrow\mathcal{D}_k$ is a $k$-valued logical function.
\end{enumerate}
\end{definition}

For the myopic best response adjustment rule, the strategy choice of each player at next time is the best response against the strategies of his neighbors at time $t$, say,
$$BR_i:=argmax_{x_i\in S_i}p_i(x_i,x_j(t)|j\in\mathcal{N}_i),$$
where $p_i(x_i,x_j(t)|j\in\mathcal{N}_i)=\sum\limits_{j\in\mathcal{N}_i}\frac{a_{ji}}{\sum\limits_{l\in N}a_{li}}V_r^T(C)x_ix_j(t)$ and $V_r^T(C)=(c_{11},c_{12},\cdots,c_{1k},\cdots,c_{k1},c_{k2},\cdots,c_{kk})$.
If $|BR_i|>1$, we choose the only one in the order of priority defined below
$x_i(t+1)=\mathrm{min}\{x|x\in BR_i\}$.

Represent each strategy $x_i\in S_i$ as an equivalent vector form $\delta_k^i$, then $\mathcal{D}_k\sim\Delta_k$. We have $x_i(t)\in\Delta_k, i=1,2,\cdots,n$. Furthermore, for any given profile $x(t)=(x_1(t),x_2(t),\cdots,x_n(t))\in S=\{(x_1,x_2,\cdots,x_n)|x_i\in S_i, i\in N\}:=\{s^1,s^2,\cdots,s^{k^n}\}$ can also be represented as a $k^n$-dimensional vector $x(t)=\ltimes_{i=1}^nx_i(t)\in\Delta_{k^n}$.

By Lemma \ref{2.2.3}, the evolutionary equation of the WNEG is
\begin{equation}\label{3.4}
x(t+1)=Fx(t),
\end{equation}
where $F\in\mathcal{L}_{k^n\times k^n}$ is the profile transition matrix. Its dimension grows exponentially with the number of players $n$, although the matrix $F$ is a general mathematical tool to analyze the dynamics of the WNEG. The problem is that the computational complexity becomes very high when dealing with matrix $F$.
Then, we introduce the definition of node aggregation to reduce the dimension of system \eqref{3.4} with dynamics invariant.

\begin{definition}(\hskip -0.1mm\cite{K16}) A node aggregation of the network $(N,A)$ associated with the WNEG is a partition $\mathcal{H}=\{H_1,H_2,\cdots,H_M\}$ of $N$, where $M\in[1:n-1]$, that is $N=H_1\cup H_2\cup\cdots\cup H_M$ and $H_i\cap H_j=\emptyset$ for all $i\neq j, i,j\in[1:M]$. $H_i,i\in[1:M]$ is an equivalence class. $\mathcal{N}_a=[1:M]$ is the super player set and $H_i,i\in[1:M]$ is a super player $i$.
$A_a:=(\bar{a}_{ij})_{M\times M}$ is the adjacency matrix after aggregation, where $\bar{a}_{ij}=\mathrm{min}\{\sum\limits_{p\in H_i}a_{pq}|\forall q\in H_j\}\geq0$. There exists an edge between super player $i$ and super player $j$ if and only if $\bar{a}_{ij}>0$. Then, the network $(N_a,A_a)$ is the aggregated network.
\end{definition}

We refer to $((N_a,A_a),G,f)$ as the aggregated WNEG.
Denote $\mathcal{N}_i^{a}=\{j|\bar{a}_{ji}>0\}, i\in[1:M]$ as the set of neighbors of super player $i$. Let $H_j=\{j_1,j_2,\cdots,j_{n_j}\}, j\in[1:M]$.
In the following, the equivalence of dynamics for evolutionary game on the network $(N,A)$ and on the aggregated network $(N_a,A_a)$ is defined.
\begin{definition}\label{9.2} For two systems \eqref{1.1} and
\begin{equation}\label{7.2}
\hat{x}_j(t+1)=f_j(\hat{x}_i(t)|i\in\mathcal{N}_j^{a}), j\in[1:M].
\end{equation}
The systems \eqref{1.1} and \eqref{7.2} are said to be equivalent if
$$x_i(t)=\hat{x}_j(t), i=j_1,j_2,\cdots,j_{n_j},j\in[1:M]; t\geq0$$
given the same initial condition $x_i(0)=\hat{x}_j(0)$.\end{definition}

If system \eqref{7.2} is equivalent to system \eqref{1.1}, then a $k^M\times k^M$-dimensional matrix can replace the matrix $F$ to determine the dynamics of the WNEG, thereby achieving the goal of dimension reduction. Therefore, in this note, we seek to explore the following question:

\textbf{Question:} What kind of network connection $(N,A)$ ensures that there exists a partition $\mathcal{H}$ in WNEG $((N,A),G,f)$ such that system \eqref{1.1} is equivalent to system \eqref{7.2} for any game $G$ by the myopic best response adjustment rule $f$.
\section{Dimension reduction of the WNEG via backward equivalence}

We first give a necessary and sufficient condition for checking whether $\mathcal{H}$ is backward equivalent. Consequently, the network structure that can be aggregated is identified and the evolutionary dynamics of the WNEG are equivalent to a low-dimension system. Next, the criteria are established to check whether $\mathcal{H}$ remains backward equivalent when adding control players as pseudo-neighbor players. This provides the reduction schemes for controlling the WNEG.
\subsection{ Backward equivalence}
Backward equivalence is given below.
\begin{definition}\label{9.4} $\mathcal{H}$ is a backward equivalence, if $x_i(t)=x_j(t), t\geq0$ holds for arbitrary $H\in\mathcal{H}$ and $i,j\in H$.
\end{definition}

From Definition \ref{9.2}, the condition that system \eqref{1.1} is equivalent to system \eqref{7.2} can be interpreted as backward equivalence. Therefore, the question we want to study is transformed to under what conditions there exists a backward equivalent partition $\mathcal{H}$ in the WNEG.
Denote $S_i=\{1,2,\cdots,k\}:=\{x_i^1,x_i^2,\cdots,x_i^k\}, i\in N$ and $x_{-i}(t)=(x_1(t),\cdots,x_{i-1}(t),x_{i+1}(t),\cdots,x_n(t))$.
Denote $S'=\{(x_1,x_2,\cdots,x_n)|x_i=x_j, \forall H\in\mathcal{H}, i,j\in H\}:=\{s^{\rho_1},s^{\rho_2},\cdots,s^{\rho_{k^M}}\}$. Denote $S_{\eta}^q$ as the set of equivalence classes for choosing strategy $q$ in the $\eta-$th profile in $S'$. Then, we have the following result.
\begin{theorem}\label{3.3} $\mathcal{H}$ is a backward equivalence
if and only if for arbitrary $H,H'\in \mathcal{H}$ and $i,j\in H$, it holds that
\begin{enumerate}
\item[$\mathrm{i)}$] $x_i(0)=x_j(0)$;
\item[$\mathrm{ii)}$] $\frac{\sum\limits_{p\in H'}a_{pi}}{\sum\limits_{l\in N}a_{li}}=\frac{\sum\limits_{p\in H'}a_{pj}}{\sum\limits_{l\in N}a_{lj}}$.
\end{enumerate}
\end{theorem}
The proof can be found in the Appendix A, and here we only provide a sketch of the proof. To prove the sufficiency, we take iterative recursion from initial values.
Based on condition i), we give any profile at time $t$ where players within the same equivalence class take identical strategies. For arbitrary $H,H'\in\mathcal{H}$ and $i,j\in H$, if condition ii) and $x_i=x_j$ hold, then  $p_i(x_i,x_{-i}(t))=p_j(x_j,x_{-j}(t))$. It follows that $x_i(t+1)=x_j(t+1)$ according to the myopic best response adjustment rule.
By continuously iterating over time $t$, it is obtained that the strategies of any two players in the same equivalence class are identical all the time. This implies that $\mathcal{H}$ is a backward equivalence.

To prove the necessity, we first examine all the cases where any two players in the same equivalence class take distinct strategies. We find that if condition ii) in Theorem \ref{3.3} does not hold, then $\mathcal{H}$ is not  backward equivalent. It implies that if $\mathcal{H}$ is backward equivalent, condition ii) in Theorem \ref{3.3} must hold. The following presents key steps of necessity of condition ii).

Given any profile $s^{\rho_\theta}=x(t)=(x_1(t),x_2(t),\cdots,x_n(t))\in S', \theta=1,2,\cdots,k^M$.
For arbitrary $H\in\mathcal{H}$ and $i,j\in H$, we assume that the best strategy for player $i$ at time $t+1$ is $x_i^\xi$, and the payoff for player $j$ choosing strategy $x_j^\beta$ is greater than that of choosing strategy $x_j^\xi$. Here we only discuss the case where $\xi<\beta$. Then one has
\begin{equation}\label{4.9}
\left\{
    \begin{array}{ll}
      p_i(x_i^\xi,x_{-i}(t))-p_i(x_i^\beta,x_{-i}(t))\geq0,  \\
      p_j(x_j^\xi,x_{-j}(t))-p_j(x_j^\beta,x_{-j}(t))<0.
    \end{array}
  \right.
\end{equation}
Next, we consider the case where the best response strategy of player $i$ is $x_i^\beta$, and the payoff for player $j$ when choosing strategy $x_j^\xi$ is greater than that when choosing strategy $x_j^\beta$. Thus
\begin{equation}\label{5.0}
\left\{
    \begin{array}{ll}
      p_i(x_i^\xi,x_{-i}(t))-p_i(x_i^\beta,x_{-i}(t))<0,  \\
      p_j(x_j^\xi,x_{-j}(t))-p_j(x_j^\beta,x_{-j}(t))\geq0.
    \end{array}
  \right.
\end{equation}
Solve inequalities \eqref{4.9} and \eqref{5.0}, and there must exist a strategy $h_\theta\in[1:k]$ such that
$$\left\{
    \begin{array}{ll}
     c_{\xi h_\theta}-c_{\beta h_\theta}\geq\sum\limits_{\alpha\in[1:k]\setminus h_\theta}\frac{\sum\limits_{H'\in S_\theta^\alpha}\sum\limits_{p\in H'}a_{pi}}{\sum\limits_{H'\in S_\theta^{h_\theta}}\sum\limits_{p\in H'}a_{pi}}(c_{\xi\alpha}-c_{\beta\alpha})\\
c_{\xi h_\theta}-c_{\beta h_\theta}<\sum\limits_{\alpha\in[1:k]\setminus h_\theta}\frac{\sum\limits_{H'\in S_\theta^\alpha}\sum\limits_{p\in H'}a_{pj}}{\sum\limits_{H'\in S_\theta^{ h_\theta}}\sum\limits_{p\in H'}a_{pj}}(c_{\xi\alpha}-c_{\beta\alpha})
    \end{array}
  \right.
$$
and
$$\left\{
    \begin{array}{ll}
     c_{\xi h_\theta}-c_{\beta h_\theta}<\sum\limits_{\alpha\in[1:k]\setminus h_\theta}\frac{\sum\limits_{H'\in S_\theta^\alpha}\sum\limits_{p\in H'}a_{pi}}{\sum\limits_{H'\in S_\theta^{h_\theta}}\sum\limits_{p\in H'}a_{pi}}(c_{\xi\alpha}-c_{\beta\alpha})\\
c_{\xi h_\theta}-c_{\beta h_\theta}\geq\sum\limits_{\alpha\in[1:k]\setminus h_\theta}\frac{\sum\limits_{H'\in S_\theta^\alpha}\sum\limits_{p\in H'}a_{pj}}{\sum\limits_{H'\in S_\theta^{h_\theta}}\sum\limits_{p\in H'}a_{pj}}(c_{\xi\alpha}-c_{\beta\alpha})
    \end{array}
  \right.
$$
If $\sum\limits_{\alpha\in[1:k]\setminus h_\theta}\frac{\sum\limits_{H'\in S_\theta^\alpha}\sum\limits_{p\in H'}a_{pj}}{\sum\limits_{H'\in S_\theta^{h_\theta}}\sum\limits_{p\in H'}a_{pj}}(c_{\xi\alpha}-c_{\beta\alpha})\neq\sum\limits_{\alpha\in[1:k]\setminus h_\theta}\frac{\sum\limits_{H'\in S_\theta^\alpha}\sum\limits_{p\in H'}a_{pi}}{\sum\limits_{H'\in S_\theta^{h_\theta}}\sum\limits_{p\in H'}a_{pi}}(c_{\xi\alpha}-c_{\beta\alpha})$, then player $i$ and player $j$ must choose different strategies. Conversely, we deduce that $\sum\limits_{\alpha\in[1:k]\setminus h_\theta}\frac{\sum\limits_{H'\in S_\theta^\alpha}\sum\limits_{p\in H'}a_{pj}}{\sum\limits_{H'\in S_\theta^{h_\theta}}\sum\limits_{p\in H'}a_{pj}}(c_{\xi\alpha}-c_{\beta\alpha})=\sum\limits_{\alpha\in[1:k]\setminus h_\theta}\frac{\sum\limits_{H'\in S_\theta^\alpha}\sum\limits_{p\in H'}a_{pi}}{\sum\limits_{H'\in S_\theta^{h_\theta}}\sum\limits_{p\in H'}a_{pi}}(c_{\xi\alpha}-c_{\beta\alpha})$ if $\mathcal{H}$ is a backward equivalence. Considering all possible strategy combinations $x_i^\xi, \xi\in[1:k]$ and $x_i^\beta, \beta\in[1:k]\setminus\{\xi\}$, we arrive at that condition ii) in Theorem \eqref{3.3} holds.

Theorem \ref{3.3} is crucial for simplifying the system to another system with low dimension. In fact, if conditions i) and ii) in Theorem \ref{3.3} hold, then the strategies of all players in any equivalence class in $\mathcal{H}$ are the same from over time. Denote the strategy of super player $j\in[1:M]$ at time $t$ by $y_j(t)$.
For any given profile $x(0)=(x_1(0),x_2(0),\cdots,x_n(0))\in S'$, there exist a profile $y(t)=(y_1(0),y_2(0),\cdots,y_M(0))$ such that $x_i(0)=y_j(0), \forall i\in H_j$.
The weighted average payoff of super players $j, j\in[1:M]$ is
$$p_j(y_j,y_{-j}(0))=\sum\limits_{p\in\mathcal{N}_j^a}\frac{\bar{a}_{pj}}{\sum\limits_{l\in[1:M]}\bar{a}_{lj}}
V_r^T(C)y_jy_p(0),$$
where $y_{-j}(t)=(y_1(t),\cdots,y_{j-1}(t),y_{j+1}(t),\cdots,y_M(t))$.
We obtain $p_j(y_j,y_{-j}(0))=p_i(x_i,x_{-i}(0))$ if $y_j=x_i$. Thus, one has $y_j(1)=x_i(1), \forall i\in H_j$.
After continuous iteration, we get $y_j(t)=x_i(t), \forall i\in H_j, t\geq0$. By Definition \ref{9.2}, system \eqref{3.4} is equivalent to
\begin{equation}\label{10.4}y(t+1)=Ly(t),\end{equation}
where $y(t)=\ltimes_{j=1}^My_j(t)$ and $L\in\mathcal{L}_{k^M\times k^M}$.

On the network $(N,A)$, a partition $\mathcal{H}$ that satisfies the conditions of backward equivalence has already been identified. As a result, the evolutionary dynamics of WNEG $((N,A),G,f)$ are driven by the evolutionary equations of super players on the aggregated network $(N_a,A_a)$.
\subsection{Control via aggregating network}
When we aim to control system \eqref{3.4} by introducing external control players, system \eqref{10.4} with low dimension is undoubtedly more advantageous thanks to its low computational complexity. Noteworthily, aggregation works here. Let us sketch the idea for control via the aggregated network (see Fig.\ref{4} for illustration). We first add control players to control super players on the aggregated network $(N_a,A_a)$. Furthermore, a reduced controlled system is established. Next, we explore a control protocol on network $(N_a,A_a)$ to meet the control target. Finally, we need to figure out the condition under which the
controlled system on the network $(N,A)$ is equivalent to the controlled system on the aggregated network $(N_a,A_a)$. Consequently, once a scheme to solve some control problem is found in the system with low dimension, it also applies to the WNEG.

For the first step,
let $X_a=\{v_1,v_2,\cdots,v_m\}\subseteq N_a$ be the set of controlled super players for the system on the aggregated network. The $m$ external control players $U_a=\{\hat{u}_1,\hat{u}_2,\cdots,\hat{u}_m\}$ act as pseudo-neighbors for the controlled super players $v_1,v_2,\cdots,v_m$. That is, there is an edge between player $v_i, i=1,2,\cdots,m$ and external control $\hat{u}_i$ with positive weight $\bar{a}_i$.
We assume that $S_{\hat{u}_1}=S_{\hat{u}_2}=\cdots=S_{\hat{u}_m}=\{1,2,\cdots,k\}$ hold. Here,
for the controlled system on the aggregated network, the set of players is $\tilde{N}_a=N_a\cup U_a$. The adjacency matrix is $\tilde{A}_a=(\hat{a}_{ij})_{(M+m)\times(M+m)}$, where
$$\hat{a}_{ij}=\left\{
  \begin{array}{ll}
   \bar{a}_{ij}, & i,j\in N_a; \\
    \bar{a}_o, &  i=v_o, j=\hat{u}_o\;\mathrm{or}\; i= \hat{u}_o, j=v_o, o\in[1:m]; \\
    0, &  i,j\in U_a. \\
  \end{array}
\right.$$
The FNG remains $G$.
Players in $N_a$ update their strategies based on the myopic best response adjustment rule $f$. The strategies for players in $U_a$ are what we need to design, and this rule is denoted as $\hat{f}$. Then, the strategy updating rule is $\tilde{f}=(f,\hat{f})$.
The aggregated WNEG $((N_a,A_a),G,f)$ becomes aggregated controlled WNEG $((\tilde{N}_a,\tilde{A}_a),G,\tilde{f})$ as long as the control players are introduced. The control problem here is to figure out $\hat{f}, U_a$ and $\tilde{A}_a$ such that the control goal is achieved.
The strategy evolution equation for each super player in the aggregated controlled WNEG is
\begin{equation}\label{3.6}
\left\{
    \begin{array}{ll}
      y_i(t+1)=f_i(y_p(t)|p\in\mathcal{N}_i^a),\;i\notin X_a \\
      y_{v_i}(t+1)=f_{v_i}(\hat{u}_i(t),y_p(t)|p\in\mathcal{N}_{v_i}^a),\;v_i\in X_a.
    \end{array}
  \right.
\end{equation}
The algebraic form of \eqref{3.6} is
\begin{equation}\label{3.5}
y(t+1)=\bar{L}\hat{u}(t)y(t),
\end{equation}
where $\hat{u}(t)=\ltimes_{i=1}^{m}\hat{u}_i(t)\in\Delta_{k^m}, y(t)=\ltimes_{j=1}^My_j(t)\in\Delta_{k^M}$ and $\bar{L}\in\mathcal{L}_{k^M\times k^{M+m}}$.

For the second step, we have system \eqref{3.5}, which has a much lower dimension than the original system (see below for details). To solve the given control problems, we resort to previous control protocols.

When the control goal is achieved based on system \eqref{3.5}, let us focus on the third step, which is crucial for our analysis.
Concerning how to transform the control protocol on the aggregated network to that of the original system, we give a sketch.
Denote $X$ and $U$ as the set of controlled players and the set of control players in the network $(N,A)$.
If $\bigcup\limits_{i=1}^mH_{v_i}=\{i_1,i_2,\cdots,i_{m'}\}$, then
we add control players $u_1,u_2,\cdots,u_{m'}$ to control players $i_1, i_2, \cdots i_{m'}$ in the network $(N,A)$. That is, $X=\{i_1, i_2, \cdots i_{m'}\}$ and $U=\{u_1,u_2,\cdots,u_{m'}\}$. We also assume that $S_{u_1}=S_{u_2}=\cdots=S_{u_{m'}}=\{1,2,\cdots,k\}$ hold.
Let the weight of edge between control player $u_q$ and controlled player $i_q$ be $a_q=\frac{\sum\limits_{l\in N}a_{li_q}}{\sum\limits_ {l\in[1:M]}\bar{a}_{lv_\kappa}}\bar{a}_{\kappa}, \forall i_q\in H_{v_\kappa}, \kappa\in[1:m]$. Similarly, we have the controlled WNEG $((\tilde{N},\tilde{A}),G,\tilde{f})$, where $\tilde{N}=N\cup U, \tilde{A}=(\tilde{a}_{ij})_{(n+m')\times(n+m')}$ and $$\tilde{a}_{ij}=\left\{
  \begin{array}{ll}
   a_{ij}, & i,j\in N; \\
    a_o, &  i=i_o, j=u_o\;\mathrm{or}\; i= u_o, j=i_o, o\in[1:m']; \\
    0, &  i,j\in U. \\
  \end{array}
\right.$$
The strategy evolution equation of each player in controlled WNEG is
\begin{equation}\label{10.2}
\left\{
    \begin{array}{ll}
      x_i(t+1)=f_i(x_p(t)|p\in\mathcal{N}_i),\;i\notin X  \\
      x_{i_q}(t+1)=f_{i_q}(u_j(t),x_{i_q}(t)|p\in\mathcal{N}_{i_q}),\;i_q\in X.
    \end{array}
  \right.
\end{equation}
The algebraic form of \eqref{10.2} is
\begin{equation}\label{10.3}
x(t+1)=\bar{F}u(t)x(t),
\end{equation}
where $u(t)=\ltimes_{i=1}^{m'}u_i(t)\in\Delta_{k^{m'}}, x(t)=\ltimes_{i=1}^nx_i(t)\in\Delta_{k^n}$ and $\bar{F}\in\mathcal{L}_{k^n\times k^{n+m'}}$.

In the following, the equivalence of controlled systems on the aggregated network $(N_a,A_a)$ and the network $(N,A)$ is defined.
\begin{definition} \label{9.3} For systems \eqref{10.2}
 and
\begin{equation}\label{7.4}
\left\{
    \begin{array}{ll}
      \hat{x}_j(t+1)=f_j(\hat{x}_p(t)|p\in\mathcal{N}_i^a),\;j\notin X^a  \\
      \hat{x}_{v_q}(t+1)=f_{v_q}(\hat{u}_q(t),\hat{x}_p(t)|p\in\mathcal{N}_{v_q}^a),\;v_q\in X^a.
    \end{array}
  \right.
\end{equation}
System \eqref{10.2} and \eqref{7.4} are equivalent if
$$x_i(t)=\hat{x}_j(t), i=j_1,j_2,\cdots,j_{n_j},j\in[1:M]; t\geq0$$
given the same initial condition $x_i(0)=\hat{x}_j(0)$.\end{definition}

We obtain $x_i(t)=y_j(t), \forall i\in H_j, j\in[1:M]\setminus X^a, t\geq0$ since the control players only affect the controlled players. We only need to focus on the evolutionary dynamics of controlled players.
We consider the following two cases classified by whether the number of players in the equivalence class to which the controlled player belongs is one.

\textbf{Case 1:} $m'=m$.
{\prp If $u_\alpha(t)=\hat{u}_q(t), \forall i_\alpha\in H_{v_q}, \alpha\in[1:m'], q\in[1:m], t\geq0$ holds, the controlled system \eqref{10.3} is equivalent to controlled system \eqref{3.5}.}

\proof Each controlled player forms an equivalence class, controlling players $i_1,i_2,\cdots,i_{m'}$ in network $(N,A)$ is equivalent to controlling super players $j_1,j_2,\cdots,j_m$ in aggregated network $(N_a,A_a)$.
Thus, we get $x_i(t)=y_j(t), \forall i\in H_j, j\in X^a, t\geq0$. By Definition \ref{9.3}, controlled system \eqref{3.5} is equivalent to controlled system \eqref{10.3}.$\hfill\square$

\textbf{Case 2:} $m'>m$.

{\prp\label{6.1} If for arbitrary $i_\alpha,i_\beta\in H_{v_q}, q\in[1:m]$ $$\frac{a_\alpha}{\sum\limits_{l\in N}a_{li_\alpha}+a_\alpha}=\frac{a_\beta}{\sum\limits_{l\in N}a_{li_\beta}+a_\beta}$$ and $$\hat{u}_q(t)=u_\alpha(t)=u_\beta(t), t\geq0$$ hold, then the controlled system \eqref{10.3} is equivalent to controlled system \eqref{3.5}.}

{\proof The proof is similar to that of the sufficiency in Theorem \ref{3.3}.
Given any profile $x(0)=(x_1(0),x_2(0),\cdots,x_n(0))\in S'$,
for arbitrary $i_\alpha, i_\beta\in H_{v_q}, q\in[1:m]$, the weighted average payoff of player $i_\varphi, \varphi=\alpha,\beta$ is
\begin{align*}p_{i_\varphi}(x_{i_\varphi},x_{-i_\varphi}(0))=&\sum\limits_{H'\in\mathcal{H}}\frac{\sum\limits_{p\in H'}a_{pi_\psi}}{\sum\limits_{l\in N}a_{li_\varphi}+a_\varphi}V_r^T(C)x_{i_\varphi}x_p(0)\\&+\frac{a_\varphi}{\sum\limits_{l\in N}a_{li_\varphi}+a_\varphi}x_{i_\varphi}u_{\varphi}(0).\end{align*}
Then, we have $p_{i_\alpha}(x_{i_\alpha},x_{-i_\alpha}(0))
=p_{i_\beta}(x_{i_\beta},x_{-i_\beta}(0))$ if $x_{i_\alpha}=x_{i_\beta}$. This implies that $x_{i_\alpha}(1)=x_{i_\beta}(1)$. That is, for arbitrary $H\in\mathcal{H}$ and $a,b\in H$, there is $x_a(1)=x_b(1)$.
Keep repeating the procedure,
when $t\geq1$, for any profile $x(t)=(x_1(t),x_2(t),\cdots,x_n(t))\in S'$, there is $x_{i_\alpha}(t+1)=x_{i_\beta}(t+1)$. Then $x_a(t)=x_b(t), t\geq0$ holds for arbitrary $H\in\mathcal{H}$ and $a,b\in H$ in the controlled WNEG.

Let $x_i(0)=y_j(0), \forall i\in H_j, j\in[1:M]$ and $\hat{u}_q(t)=u_{\alpha}(t)=u_{\beta}(t), \forall i_\alpha, i_\beta\in H_{v_q}, q\in[1:m],t\geq0$, then the weighted average payoff of controlled super player $v_q, q\in[1:m]$ under the profile $(y_1(0),y_2(0),\cdots,y_M(0))$ is
\begin{align*}p_{v_q}(y_{v_q},y_{-v_q}(0))=&\sum\limits_{p\in\mathcal{N}_{v_q}^a}\frac{\bar{a}_{pv_q}}{\sum\limits_{l\in[1:M]}\bar{a}_{lv_q}+\bar{a}_q}
V_r^T(C)y_{v_q}y_p(0)\\
&+\frac{\bar{a}_q}{\sum\limits_{l\in[1:M]}\bar{a}_{lv_q}+\bar{a}_{q}}V_r^T(C)y_{v_q}\hat{u}_q(0) .\end{align*}
We have that $p_{v_q}(y_{v_q},y_{-v_q}(0))=p_{i_{\alpha}}(x_{i_\alpha},x_{-i_\alpha}(0))$ when $y_{v_q}=x_{i_\alpha}$. Furthermore, we get $y_{v_q}(1)=x_{i_\alpha}(1)$. We obtain $y_{v_q}(t)=x_{i_\alpha}(t), t\geq0$ by recursion.
By Definition \ref{9.3}, the controlled system \eqref{10.3} is equivalent to controlled system \eqref{3.5}.
$\hfill\square$

Noteworthily, this idea is general and works for any control problems. Here, our contribution is figuring out how to map the control protocol on aggregated system to original system. Therein, a sophisticated connection mapping between the two systems is given. In the following section, we are going to show the power of this idea by four control problems.
\section{Efficacy of aggregation in controlled WNEG}
In this section, we investigate strategy consensus, optimization, controllability, and optimal control of the WNEG via
aggregation. The complexity analysis are shown that our
aggregation greatly dilutes computational burdens.
\subsection{Strategy consensus of the WNEG via aggregation}
We consider strategy consensus defined below.
\begin{definition} A WNEG with algebraic form \eqref{10.3} achieves strategy consensus if there exist a control sequence $u:\mathbb{N}\rightarrow\Delta_{k^m}$ and an integer $t_0\geq0$ such that $$x_1(t)=x_2(t)=\cdots=x_n(t), \forall t\geq t_0$$ holds for any initial condition $x(0)\in\Delta_{k^n}$.
\end{definition}

A profile feedback control is given by
$$u_i(t)=g_i(x_1(t),x_2(t),\cdots,x_n(t)), i=1,2,\cdots,m',$$
with algebraic form
$$u(t)=Gx(t),$$
where $G\in\mathcal{L}_{k^{m'}\times k^n}$ is the profile feedback gain matrix.

We sketch a control protocol for strategy consensus via aggregation: Firstly, we check whether there exists a partition $\mathcal{H}$ of $N$ satisfying the second condition of Theorem \ref{3.3}. If $\mathcal{H}$ does not exist, the network of the WNEG cannot be aggregated. Otherwise, we make players in each equivalence class play the FNG first to reach strategy consensus.
Finally, we solve the curse of dimensionality based on the game between equivalence classes. When the controlled super players on the aggregated network $(N_a,A_a)$ are identified, the corresponding controlled players on the network $(N,A)$ are found to reach the same goal of control.

If the network can be aggregated, the partition $\mathcal{H}$ satisfying the second condition of Theorem \ref{3.3} is identified. If the profile set is $S$, there are profiles such that the initial strategies of the players within equivalence class are not the same. To address the non-consensus within equivalence classes, we make each equivalence class achieve strategy consensus by the game between players within the equivalence class.
For every equivalence class $H_j, j\in[1:M]$, we denote the set of profiles for the players choosing the same strategy as
$\Theta_j=\{(\delta_{k}^i)^{n_j}|i=1,2,\cdots,k\}$.
We introduce control invariant subset.
\begin{definition} (\hskip -0.1mm\cite{Y19})  A subset $\Theta\subseteq\Delta_{k^n}$ is a control invariant subset of system \eqref{10.3} if for any $x(0)\in\Theta$, there exists a control sequence $u: \mathbb{N}\rightarrow\Delta_{k^m}$ such that $x(t)\in\Theta, \forall t\geq0$. A control invariant subset $I_c(\Theta)$ is the largest control invariant subset if each control invariant subset in $\Theta$ is a subset of it.\end{definition}

When control players are introduced into each equivalence class,
its evolutionary dynamics are a controlled system itself. Whether equivalence class $H_j, j\in[1:M]$ achieves strategy consensus can be verified by checking whether the controlled system converges to $I_c(\Theta_j)$ \cite{D17,H18}.

If equivalent class $H_j, j\in[1:M]$ reaches strategy consensus, then the time $t_j, j\in[1:M]$ to reach strategy consensus must be finite \cite{D17,H18}. Denote $t^*=\mathrm{max}\{t_j|j\in[1:M]\} $ as the initial time for games among equivalence classes. Since $\mathcal{H}$ satisfies the second condition of Theorem \ref{3.3} and the strategy of each player within the same equivalence class is the same at time $t, t\geq t^*$, then $\mathcal{H}$ is a backward equivalence.
Consequently, we search for a control protocol that enables the WNEG to achieve strategy consensus through system \eqref{3.5}.
The strategy set of super player $j$ is $I_c(\Theta_j)$, we denote  $I_c(\Theta_j)=\{(\delta_{k}^{j^1})^{n_j},(\delta_{k}^{j^2})^{n_j},\cdots,
(\delta_{k}^{j^{\alpha_j}})^{n_j}\}$.
Then, the set of profiles is
$\hat{S}=\{\delta_{k^{M}}^{\gamma}|\delta_{k^{M}}^{\gamma}=\ltimes_{j=1}^{M}\delta_{k}^{j^\sigma}, \sigma\in[1:\alpha_j]\}
:=\{\delta_{k^M}^{\beta_1},\delta_{k^M}^{\beta_2},\cdots,\delta_{k^M}^{\beta_\alpha}\},
 \alpha=\prod\limits_{i=1}^M\alpha_i$. Denote
$\hat{\Theta}=\bigcap\limits_{j=1}^M\{\delta_k^{j^1},\delta_k^{j^2},\cdots,\delta_k^{j^{\alpha_j}}\}$.
If $\hat{\Theta}\neq\emptyset$, we denote $\hat{\Theta}=\{\delta_k^{\mu_1},\delta_k^{\mu_2},\cdots,\delta_k^{\mu_a}\}$.
The set of profiles is
$\bar{\Theta}=\{(\delta_k^i)^M|i=\mu_1,\mu_2,\ldots,\mu_{a}\}
:=\{\delta_{k^M}^{\theta_1},\delta_{k^M}^{\theta_2},\cdots,\delta_{k^M}^{\theta_a}\}$
if all super players choose the same strategy.
For system \eqref{3.5},
the profile feedback control is $\hat{u}(t)=\hat{G}y(t),$
where $\bar{G}\in\mathcal{L}_{k^{m}\times k^{M}}$.
Split $\bar{L}$ into $k^m$ equal blocks as
$\bar{L}=[\bar{L}_1,\bar{L}_2,\cdots,\bar{L}_{k^m}]$,
where $\bar{L}_i\in\mathcal{L}_{k^M\times k^M}$ is profile transition matrix under the $i$-th control. Let $\hat{L}=\sum\limits_{i=1}^{k^m}\bar{L}_i$.
If the largest invariant subset $I_c(\bar{\Theta})$ in $\bar{\Theta}$ of system \eqref{3.5} is non-empty, we denote $I_c(\bar{\Theta})=\{\delta_{k^M}^{\lambda_1},\delta_{k^M}^{\lambda_2},\cdots
\delta_{k^M}^{\lambda_q}\}$ and take $\Xi=\sum_{i=1}^q\delta_{k^M}^{\lambda_i}$. Then, the WNEG achieve strategy consensus if and only if system \eqref{3.5} converges to $I_c(\bar{\Theta})$.
The criterion to verify the strategy consensus of the WNEG is hereby provided.
\begin{theorem}\label{6.2}(Algorithm \ref{7.6}) The WNEG with algebraic form \eqref{3.5} achieves strategy consensus if and only if
\begin{enumerate}
\item[$\mathrm{i)}$] $\hat{\Theta}\neq\emptyset$;
\item[$\mathrm{ii)}$] $I_c(\bar{\Theta})\neq\emptyset$;
\item[$\mathrm{iii)}$] there exists an integer $t_1\in[1:\alpha-q]$ such that for any initial profile $\delta_{k^{M}}^\omega\in\hat{S}$, $\Xi^T\mathrm{Col}_\omega(\hat{L}^{t_1})>0$.
\end{enumerate}
\end{theorem}
The proof can be found in the Appendix A.
The Algorithm based on Theorem \ref{6.2} is in Appendix B. If conditions i), ii) and iii) in Theorem \ref{6.2} hold, we then have the profile feedback control to ensure that the WNEG with algebraic form \eqref{3.5} achieves strategy consensus (Appendix C). Based on the idea in Subsection B of Section 4, we map the control protocol obtained from the aggregated network $(N_a,A_a)$ to the network $(N,A)$. As a result, the WNEG with algebraic form \eqref{10.3} achieves strategy consensus.


\begin{theorem} \label{10.1} For system \eqref{3.5}, the time complexity based on Algorithm \ref{7.6} decreases by $k^{3(n-M)}$ compared with system \eqref{10.3}, where $n$ is the number of players and $M$ is the number of equivalence classes. The space complexity based on Algorithm \ref{7.6} decreases by $k^{2(n-M)+m'-m}$ compared with system \eqref{10.3}.\end{theorem}

{\proof Step 1 of the algorithm first divides the matrix $\bar{L}$
into $k^m$ equal blocks, and then sums $k^m$ matrices of dimension $k^M\times k^M$. Matrix partitioning does not involve computation and only requires $O(1)$ operations. Matrix addition requires $O(k^{2M+m})$ operations, so the time complexity of Step 1 is $O(k^{2M+m})$.
Computing the one-step reachable set of profile $\delta_{k^M}^\omega$ only requires traversing the $|R_t(\bar{\Theta})|$ elements in the $\omega$-th column of the matrix $\hat{L}$, which takes $O(|R_t(\bar{\Theta})|)$ operations. Then the time required to compute the one-step reachable set of all profiles in $R_t(\bar{\Theta})$ is $O(|R_t(\bar{\Theta})|^2)$. Calculating $R_{t+1}(\bar{\Theta})$ requires $O(\sum\limits_{\delta_{k^M}^\omega\in R_{t}(\bar{\Theta})}|R_{t+1}(\delta_{k^M}^{\omega})|)$ operations. Therefore, the time complexity of Step 4 is $O(|R_t(\bar{\Theta})|^2+\sum\limits_{\delta_{k^M}^\omega\in R_{t}(\bar{\Theta})}|R_{t+1}(\delta_{k^M}^{\omega})|)$.
Step 5 requires $O(|R_t(\bar{\Theta})|+|R_{t+1}(\bar{\Theta})|)$ operations to check if $R_{t+1}(\bar{\Theta})\subseteq R_t(\bar{\Theta})$. The while loop is executed $t_0$ times, then the total time complexity is $O(\sum\limits_{t=0}^{t_0-1}|R_t(\bar{\Theta})|^2+\sum\limits_{\delta_{k^M}^\omega\in R_{t}(\bar{\Theta})}|R_{t+1}(\delta_{k^M}^{\omega})|)
+O(\sum\limits_{t=0}^{t_0-1}|R_t(\bar{\Theta})|+|R_{t+1}(\bar{\Theta})|)$.
Step 9 takes $O(qk^M)$ operations to compute the sum of $q$ profiles. In Step 10 of algorithm, the calculation of $\bar{L}^{\alpha-q}$ and $\Xi^T\mathrm{Col}_{\omega}(\hat{L}^{\alpha-q}), \delta_{k^{M}}^\omega\in \hat{S}\setminus\mathcal{R}_{t_0}(\bar{\Theta})$ requires $O(k^{3M}+(\alpha-q)k^M)$ operations. Step 11 is to check if $\Xi^T\mathrm{Col}_{\omega}(\hat{L}^{\alpha-q})>0, \delta_{k^{M}}^\omega\in \hat{S}\setminus\mathcal{R}_{t_0}(\bar{\Theta})$, then it requires $O(\alpha-q)$ operations. To summarize, the time complexity of Algorithm \ref{7.6} is $O(k^{3M})$.
The space complexity of the algorithm is mainly reflected in the storage of the matrices and sets.
For matrix $\bar{L}$, the space required to store all elements is
$O(k^{2M+m})$. Since $\bar{L}$ is a logical matrix, it is possible to store only the positions of the non-zero elements, then the space complexity is optimized to $O(k^{m+M})$. Similarly, the total space to store $\bar{L}_\tau, \tau\in[1:k^m]$ is either $O(k^{2M+m})$ or $O(k^{m+M})$. However, matrix $\hat{L}$ is not a logical matrix, the storage space required is $O(k^{2M})$. The storage space required for all sets during the while loop is $O((\sum\limits_{t=0}^{t_0}|\mathcal{R}_t(\bar{\Theta})|+\sum\limits_{t=0}^{t_0-1}\sum\limits_{\delta_{k^M}^\omega\in R_{t}(\bar{\Theta})}|R_{t+1}(\delta_{k^M}^{\omega})|)k^M)$.
The space required for matrix $\hat{L}^{\alpha-q}$ is $O(k^{2M})$. Thus, the space complexity of Algorithm \ref{7.6} is less than or equal to $O(k^{2M+m})$.

If the aggregation method is not used during the game between equivalence classes, then the time complexity for verifying whether the WNEG achieves strategy consensus is $O(k^{3n})$, and the maximum space complexity is $O(k^{2n+m'})$. Evidently, after aggregation, the time complexity decreases by $k^{3(n-M)}$, and the space complexity reduces by $k^{2(n-M)+m'-m}$. Furthermore, the reduction effect becomes more pronounced as the number of strategies increases.$\hfill\square$
}

We provide an illustrative example to show how the WNEG achieves strategy consensus via aggregation.
\begin{figure}[h]
  \center
  \scriptsize
\includegraphics[width=9cm,height=14cm]{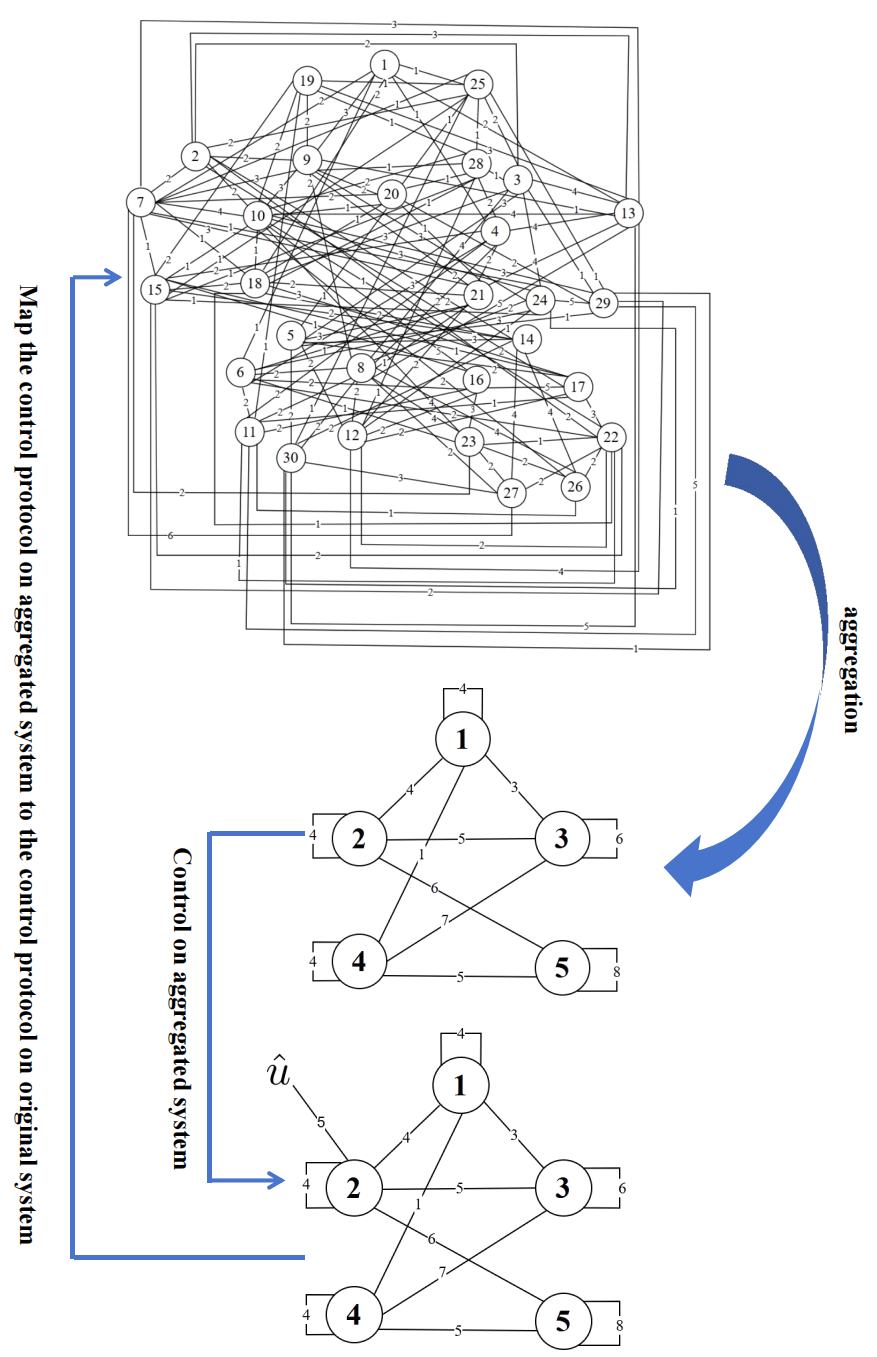}
  \caption{Illustration of aggregation. $H_1=\{1,9,19,20,25,28\}, H_2=\{2,7,10,15,18\},H_3=\{3,4,13,21,24,29\},H_4=\{5,6,8,11,12,30\}
$ and $H_5=\{14,16,17,22,23,26,27\}$ become super players 1,\;2,\;3,\;4 and 5, respectively.}\label{4}
  \vspace{-0.5em}
\end{figure}
Consider a WNEG, whose player set is $N=[1:30]$ and the strategy set of each player is $S_i=\{1,2\}, i\in N$. The network for the thirty players is in Fig. \ref{4}.
The payoff matrix is shown in Table \ref{7.8}. Players update strategies based on the myopic best response.

\begin{table}[h]
\caption{Payoff bi-matrix}
\begin{center}
\begin{tabular}{c|c|c}
\hline
\multicolumn{1}{c|}{$Player~1\backslash Player~2$}
&\multicolumn{1}{|c|}{1}
&\multicolumn{1}{|c}{2}\\ \hline
1&(0.8, 0.8) &(0.5, 0.6) \\ \hline
2 &(0.6, 0.5) &(1, 1) \\ \hline
\end{tabular}
\end{center}
\label{7.8}
\end{table}

First, a partition $\mathcal{H}=\{H_1,H_2,H_3,H_4,H_5,\}$ of $N$ is given, where $H_1=\{1,9,19,20,25,28\}, H_2=\{2,7,10,15,18\},H_3=\{3,4,13,21,24,29\},H_4=\{5,6,8,11,12,30\}
$ and $H_5=\{14,16,17,22,23,26,27\}$. For any $i,j\in H_\omega, \omega\in[1:5]$, we have \begin{equation}\label{8.2}\frac{\sum\limits_{p\in H_q}a_{pi}}{\sum\limits_{l\in N}a_{li}}=\frac{\sum\limits_{p\in H_q}a_{pj}}{\sum\limits_{l\in N}a_{lj}}, q\in[1:5].\end{equation}
The values of \eqref{8.2} for $\omega$ and $q$ are shown in Table \ref{8.1}.
\begin{table}[htb]
\caption{The values of \eqref{8.2} at different $\omega$ and $q$.}
\vskip 2mm
\begin{center}
\begin{tabular}{|c|c|c|c|c|c|}
  \hline
    \diagbox{$q$}{$\omega$}  &1 & 2 & 3&4&5 \\\hline
  1 & $\frac{1}{3}$ & $\frac{4}{19}$ & $\frac{1}{7}$&$\frac{1}{17}$&0 \\\hline
  2 &$\frac{1}{3}$&$\frac{4}{19}$&$\frac{5}{21}$&0  & $\frac{6}{19}$ \\\hline
  3 & $\frac{1}{4}$ &$\frac{5}{19}$& $\frac{2}{7}$& $\frac{7}{17}$& 0\\\hline
  4 & $\frac{1}{12}$& 0 & $\frac{1}{3}$&$\frac{4}{17}$&$\frac{5}{19}$ \\\hline
  5 &0& $\frac{6}{19}$ & 0&$\frac{5}{17}$&$\frac{8}{19}$ \\\hline
\end{tabular}
\end{center}
\label{8.1}
\end{table}

Then, we consider the strategy consensus of the WNEG. In the first stage, the strategy consensus of each equivalence class is investigated.
For the equivalence classes $H_1,H_2,\cdots,H_5$, we assume that the controlled players are 1,\;15,\;24,\;30,\;14. That is, there exist five edges between the control players $u^1,u^2,u^3,u^4,u^5$ and the controlled players 1,\;15,\;24,\;30,\;14 with weights of 2,\;3,\;4,\;3,\;4. Identify the strategy $1\sim\delta_2^1$ and $2\sim\delta_2^2$. The evolutionary equation of equivalence class $H_1$ is
\begin{equation}
x^1(t+1)=\tilde{F}_1u^1(t)x^1(t),
\end{equation}
where $x^1(t)=x_1(t)x_9(t)x_{19}(t)x_{20}(t)x_{25}(t)x_{28}(t), u^1(t)\in\Delta_2$ and $\tilde{F}_1=\delta_{64}[1,1,1,13,\cdots,64,64,64,64]\in\mathcal{L}_{64\times128}$.

All the players in $H_1$ either choose strategy `1' or strategy `2' when achieving strategy consensus, then $\Theta_1=\{\delta_{64}^1,\delta_{64}^{64}\}$. Split $\tilde{F}_1$ into 2 equal block as $\tilde{F}_1=[\tilde{F}_{1,1},\tilde{F}_{1,2}]$ and compute $\hat{F}_1=\tilde{F}_{1,1}+\tilde{F}_{1,2}$.
Since $\mathcal{R}_1(\Theta_1)=\Theta_1$, then $I_c(\Theta_1)=\Theta_1$. Let $\Phi_1=\delta_{64}^1+\delta_{64}^{64}$.
By Theorem \ref{6.2}, it is found that $\Phi_1^{T}\mathrm{Col}_i(\hat{F}_1^3)>0, \forall i\in[1:64]$. Then, the equivalence class $H_1$ reaches strategy consensus at time $t=3 $.
Similarly, we obtain that equivalence classes $H_2,H_3,H_4,H_5$ achieve consensus in strategy `1' or `2' at time $t=2,t=5,t=3,t=4$, respectively.

Since $\mathrm{max}\{2,3,4,5\}=5$, then all equivalence classes reach strategy consensus at time $t=5$.
When $t>5$, players start the game between equivalence classes using the strategy at the time $t = 5$ as the initial strategy. Because the initial strategies of the players within each equivalence class are the same and the partition $\mathcal{H}$ satisfies the second condition of Theorem \ref{3.3}, $\mathcal{H}$ is backward equivalent. Then, the network of the WNEG can be aggregated as shown in Fig. \ref{4}. Assume that the external control player set is $U=\{u_1,u_2,u_3,u_4,u_5\}$, which control players 2,\;7,\;10,\;15,\;18, respectively. We let the weights of the edges between the control players $u_1,u_4,u_5$ and controlled players 2,\;15,\;18 to be 5 and the weights of the edges between the control players $u_2,u_3$ and controlled  players 7,\;10 to be 10. Suppose that $u_1(t)=\cdots=u_5(t), t\geq5$. Then, in the aggregated network, we control the super player 2 and let the weight of the edge between $\hat{u}$ and the super player 2 be 5. By Proposition \ref{6.1}, we let $\hat{u}(t)=u_1(t)=\cdots=u_5(t), t\geq5$. Then, the evolutionary dynamics of controlled WNEG are equivalent to
\begin{equation}\label{10.5}
y(t+1)=\bar{L}\hat{u}(t)y(t),
\end{equation}
where $y(t)=\ltimes_{i=1}^5y_i(t), \hat{u}(t)\in\Delta_2$ and $\bar{L}=\delta_{32}[1,11,6,16,\cdots,32,32,32,32]\in\mathcal{L}_{32\times64}$.

Since $\hat{\Theta}=\{\delta_2^1,\delta_2^2\}$, then $\bar{\Theta}=\{\delta_{32}^1,\delta_{32}^{32}\}$. Split $\bar{L}$ into $\bar{L}=[\bar{L}_1,\bar{L}_2]$ and compute $\hat{L}=\bar{L}_1+\bar{L}_2$. We obtain $I_c(\bar{\Theta})=\bar{\Theta}$. Let $\Xi=\delta_{32}^1+\delta_{32}^{32}$, one has $\Xi^T\mathrm{Col}_i(\hat{L}^5)>0, \forall i\in[1:32]$. So, the WNEG with algebraic
form \eqref{10.5} achieves strategy consensus under the profile feedback control $\hat{u}(t)=\bar{G}y(t)$. The corresponding profile feedback gain matrix is
$$\left\{
  \begin{array}{ll}
    \mathrm{Col}_i(\bar{G})=\delta_2^1, & i\in\{1,17\}; \\
    \mathrm{Col}_i(\bar{G})=\delta_2^2, & i\in\{5,7,13,15,19,25,27\} \\
   \mathrm{Col}_i(\bar{G})\in\Delta_2, & \mathrm{otherwise}.
  \end{array}
\right.$$
Next, we map the control protocol on aggregated system to the control protocol on original system. The profile feedback gain  matrix is
$$\left\{
  \begin{array}{ll}
    \mathrm{Col}_i(G)=\delta_{32}^1, & i\in\{1,538971173\}; \\
    \mathrm{Col}_i(G)=\delta_{32}^{32}, & i\in\{201458243,256770628,\\&479367747,534680132,\\&594283558,816880677,\\&872193062\} \\
   \mathrm{Col}_i(G)\in\{\delta_{32}^1,\delta_{32}^{32}\}, & \mathrm{otherwise}.
  \end{array}
\right.$$
{\rem Based on Theorem \ref{10.1}, the time complexity of verifying the strategy consensus of the WNEG decreases by $2^{75}$, and the space complexity decreases by $2^{54}$.}
\subsection{Strategy optimization of the WNEG via aggregation}
In this section, we consider the problem of strategy optimization of the WNEG. Our goal is to ensure that the total payoff of all participants is greater than or equal to a target value $Q$ after a certain period of time by adjusting the strategies of the control players.

We sketch the control protocol for strategy optimization via
aggregation: Firstly, we check the network of the WNEG can be aggregated based on the second condition of Theorem \ref{3.3}. Secondly, if the network can be aggregated, we ensure that the initial strategies of players are the same within each equivalence class. Finally, we map the control scheme obtained from the controlled system on the aggregated network, which enables the WNEG to achieve the strategy optimization objective, to the original controlled system.

If the partition $\mathcal{H}$ of $N$ satisfying the second condition of Theorem \ref{3.3} exists, then $\mathcal{H}$ is identified. We constrain the initial profiles within $S'$, $\mathcal{H}$ is a backward equivalence. Consequently, we resort to system \eqref{3.5} to figure out a control protocol that enables the WNEG to achieve the strategy optimization goal.

Denote $B=(b_{ij})_{M\times M}$ and $b_{ij}=\frac{\bar{a}_{ij}}{\sum\limits_{l=1}^M\bar{a}_{lj}}, \forall i,j\in[1:M] $.
Denote $ \Lambda_M^i=\left\{
       \begin{array}{ll}
        I_k\otimes1_{k^{M-1}}^T, & i=1;  \\
         1_{k^{i-1}}^T\otimes I_k\otimes1_{k^{M-i}}^T, & i\in[2:M-1];  \\
         1_{k^{M-1}}^T\otimes I_k, & i=M;
       \end{array}
     \right.
$ and $\Lambda=[\Lambda_M^1,\Lambda_M^1,\cdots,\Lambda_M^M]^T$.
The weighted average payoff function of super player $i, i\in[1:M]$ is transformed into the following form:
\begin{align*}
p_i(y_i(t),y_{-i}(t))=&\sum\limits_{j=1}^{M}b_{ji}V_r^{T}(C)y_i(t)y_j(t)\\
=&V_r^{T}(C)(I_{k}\otimes\mathrm{Col}_i^{T}(B)\Lambda)\Lambda_M^iM_{r,k^M}y(t)\\
:=&V_iy(t)
\end{align*}
where $M_{r,k^M}=diag(\delta_{k^M}^1,\delta_{k^M}^2,\cdots,\delta_{k^M}^{k^M})$.
The total payoff of all players in the WNEG is $V=\sum\limits_{i=1}^M|H_i|V_i$. Denote $\Omega=\{\delta_{k^M}^i|\mathrm{Col}_i(V)\geq Q\}$.
For any profile $\delta_{k^M}^i\in\Omega$, the total payoff of all players corresponding to this profile is always greater than or equal to $Q$. Therefore, the WNEG achieving the strategy optimization goal is equivalent to system \eqref{3.5} converging to $\Omega$. To this end, we first compute the largest control invariant subset $I_c(\Omega)$ in $\Omega$ of system \eqref{3.5} \cite{Z20,Y21}.
When $I_c(\Omega)\neq\emptyset$, we construct the truth matrices  $T_{\overline{W}_{t-1}|W_t}\in\mathcal{B}_{k^m\times k^M}, t\geq0$ as follows:
$$[T_{\overline{W}_{t}|W_{t+1}}]_{ij}=\left\{
                               \begin{array}{ll}
                                 1, & \mathrm{if}\; \bar{L}\delta_{k^m}^i\delta_{k^M}^j\in \overline{W}_{t}, \forall\delta_{k^M}^j\in W_{t+1}  \\
                                 0, & \mathrm{otherwise},
                               \end{array}
                             \right.$$
where $\left\{
         \begin{array}{ll}
           \overline{W}_t=I_c(\Omega), & t=0 \\
           \overline{W}_t=\{\delta_{k^M}^j|\mathrm{Col}_j(T_{\overline{W}_{t-1}|W_t})\neq\mathbf{0}_{k^m}\}, & t\geq1
         \end{array}
       \right.$
 and $W_{t+1}=\Delta_{k^M}\setminus\bigcup\limits_{q=0}^t\overline{W}_q$.

The following result is to verify the strategy optimization of the WNEG, which comes from \cite{Z20,Y21}.
{\lem\label{6.5}(Algorithm \ref{suanfa1}) The WNEG with algebraic form \eqref{3.5} achieves the strategy optimization objective if and only if
\begin{enumerate}
\item[$\mathrm{i)}$] $I_c(\Omega)\neq\emptyset$ ;
\item[$\mathrm{ii)}$] there exist an integer $t_2\in[1:k^M-|I_c(\Omega)|]$, such that
$$\mathrm{Col}_j(\hat{T}_1)\neq\mathbf{0}_{k^m}, \forall j\in[1:k^M],$$
\end{enumerate}
where $\hat{T}_1=T_{\overline{W}_0|\overline{W}_0}+\sum\limits_{t=0}^{t_2-1}
T_{\overline{W}_{t}|W_{t+1}}$.
}
The Algorithm based on Lemma \ref{6.5} is in Appendix B. If conditions i) and ii) in Lemma \ref{6.5} hold, we then give the profile feedback control to ensure that the WNEG with algebraic form \eqref{3.5} achieves strategy optimization objective (Appendix C). Furthermore, the control protocol on the network $(N,A)$ that enables the WNEG with algebraic form \eqref{10.3} to achieve strategy optimization goal is presented.

\begin{theorem} \label{10.6} For system \eqref{3.5}, the
time complexity for verifying whether the WNEG achieves the strategy optimization objective decreases by $k^{(n-M)+m'-m}$ compared with system \eqref{10.3}. The space complexity decreases by $k^{2(n-M)+m'-m}$ compared with system \eqref{10.3}.
\end{theorem}

{\proof Before analyzing the time complexity and  space complexity of Algorithm \ref{suanfa1}, we first present the time complexity and space complexity for calculating the largest control invariant subset $I_c(\Omega)$ by referring to Algorithm 1 in \cite{Z20}. Here, the input values are $\Omega$ and $\bar{L}$, the output value is $I_c(\Omega)$. In the calculation process, $\Theta_t, t\in[0:\hat{t}]$ is replaced by $\Omega_t$. Step 3 of the algorithm is to construct a $k^m\times k^M$-dimensional truth matrix $T_{\Omega_t|\Omega_t}$. Since we only focus on the profiles within $\Omega_t$, the time complexity of Step 3 is $O(|\Omega_t|k^m)$.
Step 5 of the algorithm constructs a new set by traversing the elements of $|\Omega_t|$ columns of the matrix $T_{\Omega_t|\Omega_t}$, which requires $O(|\Omega|k^m)$ operations. Comparing whether sets $\Omega_t$ and $\Omega_{t+1}$ are equal requires $O(|\Omega_t|+|\Omega_{t+1}|)$ operations.
The total time complexity of computing $I_c(\Omega)$ is the time required for the while loop, which is $O(\sum\limits_{t=0}^{\hat{t}-1}|\Omega_t|k^m)+O(\sum\limits_{t=0}^{\hat{t}-1}|\Omega_t|k^m)+
O(\sum\limits_{t=0}^{\hat{t}-1}|\Omega_t|+|\Omega_{t+1}|)$.
Then, we analyse the space complexity of calculating $I_c(\Omega)$.
Considering different storage methods, the space required for matrix $\bar{L}$ is either $O(k^{2M+m})$ or $O(k^{m+M})$. Similarly, if the conventional storage method is used for the truth matrix $T_{\Omega_t|\Omega_t}, t\in[0:\hat{t}-1]$, it requires $O(\hat{t}k^{m+M})$ space. Since the truth matrix $T_{\Omega_t|\Omega_t}$ contains many all-zero columns, a column compression can be used for storage, which only requires $O(\hat{t}(k^m+2k^M+1))$ space.
The space required to store set $\Omega_t, t\in[0:\hat{t}]$ is $O(\sum\limits_{t=1}^{\hat{t}}|\Omega_t|k^M)$. In summary, the space complexity of this algorithm is at most $O(k^{2M+m})$.

Using the same analytical approach, steps 1-5 of Algorithm \ref{suanfa1} require $O(\bigcup\limits_{i=0}^t|\overline{W}_i|)+O(|W_{t+1}|k^m)+O(|W_{t+1}|k^m)$ operations. The time complexity of while loop is $O(\sum\limits_{t=0}^{t_2-1}\bigcup\limits_{i=0}^t|\overline{W}_i|)+O(\sum\limits_{t=0}^{t_2-1}|W_{t+1}|k^m)
+O(\sum\limits_{t=0}^{t_2-1}|W_{t+1}|k^m), t_2\in[1:k^M-|I_c(\Omega)|]$.
Step 9 of the algorithm requires $O(t_2k^{M+m})$ operations to compute the sum of $t_2$ truth matrices. Step 11 is to check whether each column of matrix $\hat{T}_1$ is non-zero, which requires $O(k^{M+m})$ operations. Therefore, the total time complexity of Algorithm \ref{suanfa1} is $O(k^{M+m})$. The space complexity of this algorithm mainly originates from the storage of sets and truth matrices. For the set $\overline{W}_t, t\in[1:t_2]$, the required storage space is $O(\sum\limits_{t=1}^{t_2}|\overline{W}_t|k^M)$. The storage space required for the truth matrices is $O(t_2k^{M+m})$ or $O(t_2(k^m+2k^M+1))$. Storing truth matrix $\hat{T}_1$ requires $O(k^{m+M})$ space.
Thus, the space complexity of Algorithm \ref{suanfa1} is $O(k^{M+m})$.
In conclusion, based on Algorithm 1 in \cite{Z20} and Algorithm \ref{suanfa1}, the time complexity for verifying whether the WNEG achieves the strategy optimization goal is $O(k^{M+m})$. The space complexity is less than or equal to $O(k^{2M+m})$.

However, if we use system \eqref{10.3} to study the strategy optimization problem of the WNEG, the time complexity is $O(k^{n+m'})$. The space complexity is less than or equal to $O(k^{2n+m'})$. Thus, for system \eqref{3.5}, the time complexity and space complexity decrease by $k^{n-M+m'-m}$ and $k^{2(n-M)+m'-m}$, respectively.
$\hfill\square$

We present an example to illustrate how to obtain a control protocol that enables the WNEG to achieve strategy optimization objective via aggregation. Consider the WNEG in the example of subsection A. We assume that the initial profile set is $S'$, $\mathcal{H}$ is a backward equivalence. The controlled players, as well as the weights of the edges between the controlled players and control players, are the same as those in the previous example.
Then, the evolutionary equation of the controlled WNEG on the aggregated network is
\begin{equation}\label{7.7}
y(t+1)=\bar{L}u(t)y(t),
\end{equation}
where  $\bar{L}=\delta_{32}[1,8,2,16,\cdots,32,32,32,32]\in\mathcal{L}_{32\times64}$.

We compute the total payoff of all players as
$$V=[24,22.8,22.06,23.38,\cdots,25.69,26.04,30]\in\mathbb{R}_{1\times32}.$$
Let $Q=24$, then $\Omega=\{\delta_{32}^i|\mathrm{Col}_i(V)\geq24\}=\{\delta_{32}^1,\delta_{32}^8,
\delta_{32}^{12},\delta_{32}^{16},\delta_{32}^{23},\delta_{32}^{24},\delta_{32}^{28},\delta_{32}^{30},\delta_{32}^{31},
\delta_{32}^{32}\}$.
We first calculate the largest control invariant subset $I_c(\Omega)$. The truth matrix $T_{\Omega|\Omega}\in\mathcal{B}_{2\times32}$ is constructed as
$$\mathrm{Col}_i(T_{\Omega|\Omega})=\left\{
  \begin{array}{ll}
    (1,1)^{T}, & i\in\{1,8,12,16,23,24,28,30,31,\\&32\}; \\
    (0,0)^T, & \mathrm{otherwise}.
  \end{array}
\right.$$
One has $\mathcal{R}_1(\Omega)=\Omega$, then $I_c(\Omega)=\Omega$.
Denote $\overline{W}_0=I_c(\Omega)$ and $W_1=\Delta_{32}\setminus\overline{W}_0$. We construct the truth matrix $T_{\overline{W}_0|W_1}\in\mathcal{B}_{2\times32}$ as
$$\mathrm{Col}_i(T_{\overline{W}_0|W_1})=\left\{
  \begin{array}{ll}
  (0,1)^T, & i\in\{11,13,25\};\\
(1,1)^{T}, & i\in\{2,4,6,7,8,10,14,15,18,\\&19,22,26,27,29\}; \\
    (0,0)^T, & \mathrm{otherwise}.
  \end{array}
\right.$$
We obtain
\begin{align*}
\overline{W}_1=\delta_{32}\{&2,4,6,7,8,10,11,13,14,15,18,19,22,25,
26,\\&27,29\}.\end{align*}
Denote $W_2=\Delta_{32}\setminus(\overline{W}_0\cup \overline{W}_1)$. We continue to construct the truth matrix $T_{\overline{W}_1|W_2}\in\mathcal{B}_{2\times32}$ as
$$\mathrm{Col}_i(T_{\overline{W}_1|W_2})=\left\{
  \begin{array}{ll}
  (0,1)^T, & i\in\{5,17\};\\
(1,1)^{T}, & i\in\{3,9,21\}; \\
    (0,0)^T, & \mathrm{otherwise}.
  \end{array}
\right.$$
We have $\overline{W}_2=\delta_{32}\{3,5,9,17,21\}$.

Since $$\mathrm{Col}_i(T_{\Omega|\Omega}+T_{\overline{W}_0|W_1}+T_{\overline{W}_1|W_2})\neq\mathbf{0}_2, \forall i\in[1:32],$$
then the WNEG reaches strategy optimization objective under the profile feedback control $\hat{u}(t)=\bar{G}y(t)$.
The profile feedback gain matrix is
$$\left\{
  \begin{array}{ll}
    \mathrm{Col}_i(\bar{G})=\delta_2^2, & i\in\{5,11,13,17,25\}; \\
   \mathrm{Col}_i(\bar{G})\in\Delta_2, & \mathrm{otherwise}.
  \end{array}
\right.$$
Then, the profile feedback control corresponding to the original controlled system is given. The profile feedback gain matrix is
$$\left\{
  \begin{array}{ll}
    \mathrm{Col}_i(G)=\delta_{32}^{32}, & i\in\{201458243,333221890,\\&479367747,538971173,\\&816880677\}; \\
   \mathrm{Col}_i(G)\in\{\delta_{32}^{1},\delta_{32}^{32}\}, & \mathrm{otherwise}.
  \end{array}
\right.$$
{\rem Based on Theorem \ref{10.6}, the time complexity of verifying the strategy optimization of the WNEG decreases by $2^{29}$, and the space complexity decreases by $2^{54}$.}

\subsection{Controllability and optimal control of the WNEG via aggregation}
The controllability defined as follows is investigated based on profile feedback control.
\begin{definition} (\hskip -0.1mm\cite{D22}) Consider system \eqref{10.3}, a initial profile $x_0\in \Delta_{k^n}$ and a destination profile  $x_d\in\Delta_{k^n}$.
\begin{enumerate}
\item[$\mathrm{i)}$] System \eqref{10.3} is said to controllable from $x_0$ to $x_d$, if there exists a integer $K>0$ and a control sequence $u:[0:K-1]\rightarrow\Delta_{k^{m'}}$ such that $x_0$ is driven by control sequence to $x_d$;
\item[$\mathrm{ii)}$] System \eqref{10.3} is said to be controllable at $x_0$, if system \eqref{10.3} is controllable from $x_0$ to any $x_d$.
\item[$\mathrm{iii)}$] System \eqref{10.3}  is said to be controllable, if system \eqref{10.3} is controllable at any $x_0$.
\end{enumerate}
\end{definition}
The outline of control protocol for controllability is the same as that for strategy optimization. That is, if the network of the WNEG can be aggregated, we identify the partition $\mathcal{H}$ that satisfies the second condition of Theorem \ref{3.3}.
Similarly, we restrict the initial profile set to $S'$, then $\mathcal{H}$ is backward equivalent. When the profile $x(0)=x_0\in S'$, we obtain $x(K)=x_d\in S'$ by Theorem \ref{3.3}.
We let $x_i(t)=y_j(t), i=j_1,j_2,\cdots,j_{n_j}, j\in[1:M], t=0,K$ and compute $y(t)=\ltimes_{j=1}^My_j(t)$. Denote $y_0=y(0)$ and $y_d=y(K)$, then system \eqref{3.5} is to explore the control scheme for controllability of the WNEG, with reference to \cite{H23}.
{\lem \label{6.7} (Algorithm \ref{7.5}) Consider system \eqref{3.5} and two profiles $y_0=\delta_{k^M}^{i}$ and $y_d=\delta_{k^M}^{j}$. Then,
\begin{enumerate}
\item[$\mathrm{i)}$] system \eqref{3.5} is controllable from $y(0)=\delta_{k^M}^{j}$ to $y(K)=\delta_{k^M}^{i}$ if and only if $$[\hat{L}^K]_{ij}>0.$$
\item[$\mathrm{ii)}$] system \eqref{3.5} is controllable at $\delta_{k^M}^{j}$ if and only if
$$\sum\limits_{K=1}^{k^M}\mathrm{Col}_j(\hat{L}^K)>0.$$
\item[$\mathrm{iii)}$] system \eqref{3.5} is controllable if and only if $$\sum\limits_{K=1}^{k^M}\hat{L}^K>0.$$
\end{enumerate}
where $\hat{L}$ is constructed in subsection A.
}

The Algorithm based on Lemma \ref{6.7} is in Appendix B.
When system \eqref{3.5} is controllable, we map the control protocol to system \eqref{10.3}, then the WNEG is controllable.
\begin{theorem} \label{10.7} For system \eqref{3.5}, the time complexity based on Algorithm \ref{7.5} decreases by $k^{4(n-M)}$ compared with system \eqref{10.3}. The space complexity based on Algorithm \ref{7.5} decreases by $k^{3(n-M)}$ compared with system \eqref{10.3}.\end{theorem}

{\proof
Step 1 of the algorithm is matrix partitioning, which requires $O(1)$ operations. The time complexity of computing matrix $\hat{L}\in\mathbb{R}_{k^M\times k^M}$ is $O(k^{2M+m})$. When $K=1$, calculating $\hat{L}^K$ requires $O(1)$ operations. When $K\geq2$, the calculation of $\hat{L}^K$ requires $O(k^{3M})$ operations.
The time complexity of Step 2 is $O(k^{2M+m}+1)$ or $O(k^{2M+m}+k^{3M})$. Step 3 requires $O(1)$ operations to check whether $[\hat{L}^K]_{ij}$ is nonzero. Step 6 first calculates $\hat{L}^\mu, \mu\in[1:k^M]$, followed by calculating $\sum\limits_{\mu=1}^{k^M}\hat{L}^\mu$, so the time complexity is $O(1+(k^M-1)k^{3M}+k^{3M})$. Step 7 requires $O(k^M)$ operations to check whether a column of the matrix is nonzero. Step 10, which is to check whether each element of the matrix is nonzero, requires $O(k^{2M})$ operations. In summary, the time complexity of Algorithm \ref{7.5} is $O(k^{4M})$. The space complexity of Algorithm \ref{7.5} mainly originates from the storage of matrices. For matrix $\bar{L}$, the storage space required by conventional storage and compressed storage are $O(k^{2M+m})$ and $O(k^{M+m})$, respectively. After dividing $\bar{L}$ into $k^m$ equal blocks, the total space required to store $k^m$ matrices of dimension $k^M\times k^M$ is either $O(k^{2M+m})$ or $O(k^{m+M})$. The space required to store matrices $\hat{L}^{\mu}, \mu\in[2:k^M]$ and $\sum\limits_{\mu=1}^{k^M}\hat{L}^{\mu}$ is $O(k^{3M})$. Therefore, the space complexity of Algorithm \ref{7.5} is $O(k^
{3M})$.

If controlled system \eqref{10.3} is used, the time complexity and space complexity for verifying whether the WNEG is controllable are $O(k^{4n})$ and $O(k^{3n})$, respectively.
After aggregation, the time complexity reduces by $k^{4(n-M)}$. The space complexity reduces by $k^{3(n-M)}$. $\hfill\square$

In the following, on the basis of research on controllability, we study the Mayer-type optimal control problem to save the control cost.
Consider system \eqref{3.5} with the initial profile $\delta_{k^M}^{j}$. The Mayer-type optimal control is to find a control sequence such that the cost functional
$$J(u(0),u(1),\cdots,u(K-1);y(0))=r^{T}y(K,u)$$
is minimized, where $r=[r_1,r_2,\cdots,r_{k^M}]^T\in\mathbb{R}_{k^M\times1}$ is a fixed vector, and $K>0$ is a fixed or designed shortest termination time.
References \cite{H23,H24} have provided specific algorithms for the optimal control problems in the cases where $K>0$ is a fixed termination time and $K>0$ is the designed shortest termination time.

When the optimal control sequence of system \eqref{3.5} is found, then we present the optimal control sequence of system \eqref{10.3}, which is equivalent to system \eqref{3.5}.
\begin{theorem} \label{10.8} For system \eqref{3.5}, when $K$ is fixed termination time,  the time complexity of designing the optimal control sequence decreases by $k^{2(n-M)}$ or $k^{3(n-M)}$ compared with system \eqref{10.3}. The space complexity decreases by $k^{2(n-M)+m'-m}$. When $K$ is the designed shortest termination time, the time complexity decreases by $k^{4(n-M)}$ compared with system \eqref{10.3}. The space complexity decreases by $k^{3(n-M)}$.\end{theorem}

\proof When $K$ is fixed termination time, we refer to Algorithm 6.2 in \cite{H23} or Algorithm \Rmnum{3}.1 in \cite{H24}. The time complexity and space complexity of the algorithm are analyzed as follows. The time complexity of calculating the matrix $\hat{L}$ is $O(k^{2M+m})$. When $K=1$,
traversing the $j$-th column of the matrix $\hat{L}$ to
calculate the reachable set $\mathcal{R}_K(\delta_{k^M}^{j})$ requires $O(k^M)$ operations. Step 2 is to determine the optimal value requires $O(|\mathcal{R}_K(\delta_{k^M}^{j})|)$ operations. Step 3 finds the control corresponding to the optimal value by determining whether an element of the matrix $\bar{L}_i, i\in[1:k^m]$ is non-zero, which requires $O(k^m)$ operations.
When $K\geq2$, the first step of the algorithm needs to compute the matrix $\hat{L}^K$. The time complexity of Step 1 is $O(k^{3M}+k^M)$. Step 2 also takes
$O(|\mathcal{R}_K(\delta_{k^M}^{j})|)$ operations. Since $K>1$, we directly execute the fourth step. Step 4 first requires $O(k^{M+m})$ operations to find the reachable set $\mathcal{R}_1(y(K))$, then verifying whether $y(K-2)$ reaches $\mathcal{R}_1(y(K))$ requires $O(|\mathcal{R}_1(y(K))|k^m)$ operations. The same steps are performed $K-1$ times to find an optimal control sequence, the time complexity of the loop is $O(K-1(k^{M+m}))$.
When $K=1$ and $K\geq2$, the total time complexity of this algorithm are $O(k^{2M+m})$ and $O(k^{3M})$, respectively. The space complexity is also no more than $O(k^{2M+m})$. Without aggregation, when $K=1$ and $K\geq2$ the time complexity based on this algorithm are $O(k^{2n+m'})$ and $O(k^{3n})$, respectively. The space complexity does not exceed $O(k^{2n+m'})$.

For system \eqref{3.5}, the time complexity of finding the optimal control sequence decreases by $k^{2(n-M)}$ and  $k^{3(n-M)}$ when $K=1$ and $K\geq2$, respectively. The space complexity reduces by $k^{2(n-M)+m'-m}$ in both cases.

When $K$ is the designed shortest termination time, we consider Algorithm 6.3 in \cite{H23} or Algorithm \Rmnum{3}.2 in \cite{H24}.
Unlike $K$ being a fixed termination time, the first step of the algorithm requires calculating all-step reachable sets of profile $\delta_{k^{M}}^{j}$.
First, it is required to perform $k^M$ matrix power operations, then sum $k^M$ matrices, and finally traverse the $j$-th column of the matrix. The total time complexity is $O(k^{4M})$. Since other steps are similar to those with $K$ as the fixed termination time, and their time complexities are all smaller than $O(k^{4M})$, then the time complexity of the algorithm is thus
$O(k^{4M})$. The space complexity is $O(k^{3M})$.  For system \eqref{10.3}, the time complexity of finding the optimal control sequence is $O(k^{4n})$ when $K$ is the designed shortest termination time. The space complexity is $O(k^{3n})$. Thus, the time complexity and space complexity decrease by $k^{4(n-M)}$ and $k^{3(n-M)}$ respectively at this point. $\hfill\square$

We propose an example to illustrate how to study the controllability and optimal control of the WNEG via aggregation.
Consider system \eqref{7.7}.
Split $\bar{L}$ into $\bar{L}=[\bar{L}_1,\bar{L}_2]$ and compute $\hat{L}=\bar{L}_1+\bar{L}_2$.
Let $y(0)=\delta_{32}^5$, we have $[\hat{L}]_{35}=1>0$ and $[\hat{L}^4]_{325}=16>0$. Then, system \eqref{3.5} is controllable from $\delta_{32}^5$ to $y(1)=\delta_{32}^3$ and $y(4)=\delta_{32}^{32}$.
Since $\delta_{32}^{32}$ is a fixed point, system \eqref{3.5} is neither controllable nor controllable at $\delta_{32}^5$. That is, the WNEG is controllable from $\delta_{2^{30}}^{201458243}$ to $y(1)=\delta_{2^{30}}^{55312386}$ and $y(4)=\delta_{2^{30}}^{2^{30}}$.
The WNEG is neither controllable nor controllable at $\delta_{2^{30}}^{201458243}$.

Assume that the cost functional is $J(u(0),u(1),\cdots,u(K-1);y(0))=r^{T}y(K,u)$,
where $r=[1,2,3,5,2,4,5,3,4,6,4,3,2,1,4,6,7,6,4,6,3,4,6,4,3,
2,\\1,4,6,7,6,2]$.
We are committed to finding a control sequence to minimize the cost functional with $y(0)=\delta_{32}^5$ and $K=3$.
We calculate $\mathcal{R}_3(\delta_{32}^5)=\{\delta_{32}^{8},\delta_{32}^{16},\delta_{32}^{24},\delta_{32}^{32}\}$.
Minimize the cost functional under the reachable set constraints and obtain that $r^*=\mathrm{min}\{2,3,4,6\}=2$ and the corresponding $y(K)=\delta_{32}^{32}$. Since $[\bar{L}_1]_{35}=1,[\bar{L}_2]_{103}=1$ and $[\bar{L}_2]_{3210}=1$, then we find a control sequence $u(0)=\delta_2^1,u(1)=\delta_2^2 $ and $u(2)=\delta_2^2$ to minimize the cost functional at time step $K=3$. For the original controlled system, when $K=3$, the control cost of
profile $\delta_{2^{30}}^{201458243}$
reaching profile $\delta_{2^{30}}^{2^{30}}$ under the control sequence $u(0)=\delta_{32}^1,u(1)=\delta_{32}^{32} $ and $u(2)=\delta_{32}^{32}$ is minimized.
{\rem Based on Theorem \ref{10.7}, the time complexity of verifying the controllability of the WNEG decreases by $2^{50}$ or $2^{75}$, and the space complexity decreases by $2^{54}$. Based on Theorem \ref{10.8}, the time complexity of finding optimal control sequence of the WNEG decreases by $2^{100}$, and the space complexity decreases by $2^{75}$.}
\section{Conclusion}
Compared with mean-field approximations and Monte Carlo simulations, STP enables a rigorous analysis of dynamic behavior in networked evolutionary games with any number of players and any size of strategy set. However, its computational complexity is high for games with  either a large number of players or a large strategy set.
We propose an approach based on backward equivalence to aggregate nodes of WNEGs. This approach effectively reduces the computational complexity.
Firstly, a necessary and sufficient condition is proposed to tell whether
the WNEG on a large scale network is reduced to the WNEG on an aggregated network.
Secondly, we give the criteria to determine whether the \emph{controlled} WNEG is reduced to an equivalent \emph{controlled} system with low dimension.
Finally, the effectiveness of the proposed aggregation method is verified.

Our work solves the curse of dimensionality arising from STP, although the aggregation has been applied to chemical reaction networks \cite{L25} and Boolean networks \cite{G27}. Our results are more general. The condition of backward aggregation in our paper is necessary and sufficient, not merely sufficient \cite{G28}.
Compared with the replicator equation in \cite{G28}, the myopic best response adjustment rule takes into account less information. That's why we choose the myopic best response adjustment rule.
In fact, the aggregation proposed in this paper is robust for all the updating rules which are determined by one-step neighbor. Therefore, it is also valid for aspiration dynamics \cite{J32,M33,Y34}.

Significant attention has focused on ensuring the evolutionary dynamics of WNEGs converge to ideal target profiles.
However, as global natural convergence of WNEGs cannot always be  guaranteed, developing control mechanisms becomes essential.
Compared with the method of regarding some players within the player set as control players \cite{D17}, the introduction of external control players in this paper leads to higher system dimension. However, the setting that there is no interaction between control players, on the basis of achieving control objectives, also simplifies the analysis to a certain extent. This advantage is obvious when compared to the input network where controls are logical variables satisfying certain logical rule \cite{D47}.
Our aggregation method has a broad application to control problems ranging from strategy consensus, strategy optimization, controllability, to optimal control whereas the aggregation method in \cite{H18} is designed for the strategy consensus problem only. Beside, this aggregation method can also be used for the system obtained via Koopman operator \cite{B37,M38}. Similar to STP, the Koopman operator transforms a nonlinear dynamical system into a linear one in a high-dimensional space. Therefore, the dimension of the resulting linear system is computationally heavy, our aggregation method can be insightful for diluting the curse of dimensionality arising from Koopman, too.

In spite of its universality of the aggregation method, there are limitations.
If players update their strategies according to the unconditional imitation rule \cite{M1,G35,Y36}, this aggregation method is no longer applicable because the neighbors' information of focal individual's neighbors is necessary.

To sum up, our results identify the type of networks that can be aggregated by backward equivalence. This sheds new light on the control of evolutionary games on large-scale networks.

\section*{Appendix}
\subsection{proof}
Proof of Theorem \ref{3.3}: (Sufficiency) Given any initial profile $x(0)=(x_1(0),x_2(0),\cdots,x_n(0))\in S'$, for arbitrary $i,j\in H, H\in\mathcal{H}$, the weighted average payoffs of players $i$ and $j$ is separately express as
\begin{align}\label{1.3}
p_i(x_i,x_{-i}(0)) =&\sum\limits_{p\in\mathcal{N}_i}\frac{a_{pi}}{\sum\limits_{l\in N}a_{li}}V_r^T(C)x_ix_p(0)\nonumber\\
=&\sum\limits_{H'\in\mathcal{H}}\frac{\sum\limits_{p\in H'}a_{pi}}{\sum\limits_{l\in N}a_{li}}V_r^T(C)x_ix_p(0)
\end{align}
and
\begin{align}\label{1.4}
p_j(x_j,x_{-j}(0))=&\sum\limits_{p\in\mathcal{N}_j}\frac{a_{pj}}{\sum\limits_{l\in N}a_{lj}}V_r^T(C)x_jx_p(0)\nonumber\\
=&\sum\limits_{H'\in\mathcal{H}}\frac{\sum\limits_{p\in H'}a_{pj}}{\sum\limits_{l\in N}a_{lj}}V_r^T(C)x_jx_p(0).
\end{align}

From \eqref{1.3} and \eqref{1.4}, whenever $x_i=x_j$ one gets $p_i(x_i,x_{-i}(0))=p_j(x_j,x_{-j}(0))$.
Then the strategies that enable player $i$ and player $j$ to obtain the maximum payoff are also the same. According to the myopic best response adjustment rule, we derive that $x_i(1)=x_j(1)$ holds for arbitrary $H\in\mathcal{H}$ and $i,j\in H$.

For any profile $x(1)=(x_1(1),x_2(1),\cdots,x_n(1))\in S'$, we also have $p_i(x_i,x_{-i}(1))=p_j(x_j,x_{-j}(1))$ when $x_i=x_j$. Furthermore, it is inferred that $x_i(2)=x_j(2)$ holds for arbitrary $H\in\mathcal{H}$ and $i,j\in H$.

Repeating the above process, when $t\geq2$, for arbitrary $H\in\mathcal{H}$ and $i,j\in H$, it holds that $p_i(x_i,x_{-i}(t))=p_j(x_j,x_{-j}(t))$ for $x_i=x_j$, then there must be $x_i(t+1)=x_j(t+1)$.

In summary, when the two conditions of Theorem \ref{3.3} are satisfied, we have $x_i(t)=x_j(t), t\geq0$ for arbitrary $H\in\mathcal{H}$ and $i,j\in H$. By Definition \ref{9.4}, it follows that $\mathcal{H}$ is a backward equivalence.

(Necessity) Since $\mathcal{H}$ is a backward equivalence, then $x_i(t)=x_j(t), t\geq0$ holds for arbitrary $H\in\mathcal{H}$ and $i,j\in H$. It follows that the first condition of the theorem is valid.

Consider arbitrary $H\in\mathcal{H}$ and $i,j\in H$. For any profile $x(t)=(x_1(t),x_2(t),\cdots,x_n(t))=s^{i_\theta}\in S', \theta=1,2,\cdots,k^M$, we compute the weighted average payoff of player $\rho, \rho=i,j$ as
\begin{align*}
p_\rho(x_\rho,x_{-\rho}(t))=&\sum\limits_{H'\in\mathcal{H}}\frac{\sum\limits_{p\in H'}a_{p\rho}}{\sum\limits_{l\in N}a_{l\rho}}V_r^T(C)x_\rho x_p(t)\nonumber\\
=&\sum\limits_{\alpha=1}^{k}\sum\limits_{H'\in S_\theta^\alpha}\frac{\sum\limits_{p\in H'}a_{p\rho}}{\sum\limits_{l\in N}a_{l\rho}}V_r^T(C)x_\rho x_p(t).
\end{align*}

First, we present all cases where player $i$ and player $j$ choose different strategies. There are $k$ possible strategies for player $i$ to choose, while player $j$ will have $k-1$ strategies available to him if he chooses a different strategy from player $i$.
Therefore, there are a total of $k^{2}-k$ situations in which players $i$ and $j$ choose different strategies.

If player $i$ chooses strategy $x_i^{\xi}, \forall\xi\in[1:k]$ at time $t+1$, then
\begin{equation}\label{2.1}
\left\{
  \begin{array}{ll}
    p_i(x_i^\xi,x_{-i}(t))- p_i(x_i^1,x_{-i}(t))>0,  \\
 \quad\vdots\\
p_i(x_i^\xi,x_{-i}(t))- p_i(x_i^{\xi-1},x_{-i}(t))>0 \\
p_i(x_i^\xi,x_{-i}(t))- p_i(x_i^{\xi+1},x_{-i}(t))\geq0 \\
\quad\vdots\\
p_i(x_i^\xi,x_{-i}(t))- p_i(x_i^{k},x_{-i}(t))\geq0.\\
  \end{array}
\right.
\end{equation}

Since the strategy chosen by player $j$ is different from that of player $i$, then there must exist at least one strategy $x_j^\beta, $ such that \begin{equation}\label{2.2}
\left\{
  \begin{array}{ll}
  p_j(x_j^\xi,x_{-j}(t))- p_j(x_j^\beta,x_{-j}(t))\leq0, & \beta\in[1:\xi-1] \\
    p_j(x_j^\xi,x_{-j}(t))- p_j(x_j^\beta,x_{-j}(t))<0, & \beta\in[\xi+1:k]
  \end{array}
\right.
\end{equation}

When $\beta\in[1:\xi-1]$, combining \eqref{2.1} with the first inequality of \eqref{2.2} leads to:
\begin{equation}\label{4.1}
\left\{
  \begin{array}{ll}
    \sum\limits_{\alpha=1}^{k}\sum\limits_{H'\in S_\theta^\alpha}\frac{\sum\limits_{p\in H'}a_{pi}}{\sum\limits_{l\in N}a_{li}}(c_{\xi\alpha}-c_{\omega\alpha})>0,\omega\in[1:\xi-1] \\
\sum\limits_{\alpha=1}^{k}\sum\limits_{H'\in S_\theta^\alpha}\frac{\sum\limits_{p\in H'}a_{pi}}{\sum\limits_{l\in N}a_{li}}(c_{\xi\alpha}-c_{\omega\alpha})\geq0,\omega\in[\xi+1:k]\\
\sum\limits_{\alpha=1}^{k}\sum\limits_{H'\in S_\theta^\alpha}\frac{\sum\limits_{p\in H'}a_{pj}}{\sum\limits_{l\in N}a_{lj}}(c_{\xi\alpha}-c_{\beta\alpha})\leq0.
  \end{array}
\right.
\end{equation}
When $\beta\in[\xi+1,k]$, combining \eqref{2.1} and the second inequality of \eqref{2.2} yields:
\begin{equation}\label{4.2}
\left\{
  \begin{array}{ll}
    \sum\limits_{\alpha=1}^{k}\sum\limits_{H'\in S_\theta^\alpha}\frac{\sum\limits_{p\in H'}a_{pi}}{\sum\limits_{l\in N}a_{li}}(c_{\xi\alpha}-c_{\omega\alpha})>0,\omega\in[1:\xi-1] \\
\sum\limits_{\alpha=1}^{k}\sum\limits_{H'\in S_\theta^\alpha}\frac{\sum\limits_{p\in H'}a_{pi}}{\sum\limits_{l\in N}a_{li}}(c_{\xi\alpha}-c_{\omega\alpha})\geq0,\omega\in[\xi+1:k]\\
\sum\limits_{\alpha=1}^{k}\sum\limits_{H'\in S_\theta^\alpha}\frac{\sum\limits_{p\in H'}a_{pj}}{\sum\limits_{l\in N}a_{lj}}(c_{\xi\alpha}-c_{\beta\alpha})<0.
  \end{array}
\right.
\end{equation}

Keep the remaining inequalities unchanged and combine the $\beta$-th inequality with the last one, solving \eqref{4.1} and \eqref{4.2} for
\begin{equation}\label{2.3}
\left\{
  \begin{array}{ll}
    \sum\limits_{\alpha=1}^{k}\sum\limits_{H'\in S_\theta^\alpha}\frac{\sum\limits_{p\in H'}a_{pi}}{\sum\limits_{l\in N}a_{li}}(c_{\xi\alpha}-c_{\omega\alpha})>0, \omega\in[1:\xi-1]\setminus\{\beta\}\\
    \sum\limits_{\alpha=1}^{k}\sum\limits_{H'\in S_\theta^\alpha}\frac{\sum\limits_{p\in H'}a_{pi}}{\sum\limits_{l\in N}a_{li}}(c_{\xi\alpha}-c_{\omega\alpha})\geq0, \omega\in[\xi+1,k]\\
c_{\xi h_\theta}-c_{\beta h_\theta}>\sum\limits_{\alpha\neq h_\theta}\frac{\sum\limits_{H'\in S_\theta^\alpha}\sum\limits_{p\in H'}a_{pi}}{\sum\limits_{H'\in S_\theta^{h_\theta}}\sum\limits_{p\in H'}a_{pi}}(c_{\xi\alpha}-c_{\beta\alpha})\\
c_{\xi h_\theta}-c_{\beta h_\theta}\leq\sum\limits_{\alpha\neq h_\theta}\frac{\sum\limits_{H'\in S_\theta^\alpha}\sum\limits_{p\in H'}a_{pj}}{\sum\limits_{H'\in S_\theta^{h_\theta}}\sum\limits_{p\in H'}a_{pj}}(c_{\xi\alpha}-c_{\beta\alpha})
  \end{array}
\right.
\end{equation}
and
\begin{equation}\label{4.5}
\left\{
  \begin{array}{ll}
    \sum\limits_{\alpha=1}^{k}\sum\limits_{H'\in S_\theta^\alpha}\frac{\sum\limits_{p\in H'}a_{pi}}{\sum\limits_{l\in N}a_{li}}(c_{\xi\alpha}-c_{1\alpha})>0, \omega\in[1:\xi-1]\\
    \sum\limits_{\alpha=1}^{k}\sum\limits_{H'\in S_\theta^\alpha}\frac{\sum\limits_{p\in H'}a_{pi}}{\sum\limits_{l\in N}a_{li}}(c_{\xi\alpha}-c_{\omega\alpha})\geq0, \omega\in[\xi+1,k]\setminus\{\beta\}\\
c_{\xi h_\theta}-c_{\beta h_\theta}\geq\sum\limits_{\alpha\neq h_\theta}\frac{\sum\limits_{H'\in S_\theta^\alpha}\sum\limits_{p\in H'}a_{pi}}{\sum\limits_{H'\in S_\theta^{h_\theta}}\sum\limits_{p\in H'}a_{pi}}(c_{\xi\alpha}-c_{\beta\alpha})\\
c_{\xi h_\theta}-c_{\beta h_\theta}<\sum\limits_{\alpha\neq h_\theta}\frac{\sum\limits_{H'\in S_\theta^\alpha}\sum\limits_{p\in H'}a_{pj}}{\sum\limits_{H'\in S_\theta^{h_\theta}}\sum\limits_{p\in H'}a_{pj}}(c_{\xi\alpha}-c_{\beta\alpha}).
  \end{array}
\right.
\end{equation}
Consider the opposite case to the above, that is, for player $i$, choosing strategy $x_i^\beta$ can obtain the maximum payoff, and for player $j$, choosing strategy $x_i^\xi$ yields a higher payoff than choosing strategy $x_i^\beta$. The sets of inequalities in this case are:
\begin{equation}\label{4.3}
\left\{
  \begin{array}{ll}
    \sum\limits_{\alpha=1}^{k}\sum\limits_{H'\in S_\theta^\alpha}\frac{\sum\limits_{p\in H'}a_{pi}}{\sum\limits_{l\in N}a_{li}}(c_{\beta\alpha}-c_{\omega\alpha})>0, \omega\in[1:\beta-1]  \\
    \sum\limits_{\alpha=1}^{k}\sum\limits_{H'\in S_\theta^\alpha}\frac{\sum\limits_{p\in H'}a_{pi}}{\sum\limits_{l\in N}a_{li}}(c_{\beta\alpha}-c_{\xi+1\alpha})\geq0, \omega\in[\beta+1:k]\setminus\{\xi\}\\
c_{\xi h_\theta}-c_{\beta h_\theta}>\sum\limits_{\alpha\neq h_\theta}\frac{\sum\limits_{H'\in S_\theta^\alpha}\sum\limits_{p\in H'}a_{pj}}{\sum\limits_{H'\in S_\theta^{h_\theta}}\sum\limits_{p\in H'}a_{pj}}(c_{\xi\alpha}-c_{\beta\alpha})\\
c_{\xi h_\theta}-c_{\beta h_\theta}\leq\sum\limits_{\alpha\neq h_\theta}\frac{\sum\limits_{H'\in S_\theta^\alpha}\sum\limits_{p\in H'}a_{pi}}{\sum\limits_{H'\in S_\theta^{h_\theta}}\sum\limits_{p\in H'}a_{pi}}(c_{\xi\alpha}-c_{\beta\alpha}).
  \end{array}
\right.
\end{equation}
and
\begin{equation}\label{4.4}
\left\{
  \begin{array}{ll}
    \sum\limits_{\alpha=1}^{k}\sum\limits_{H'\in S_\theta^\alpha}\frac{\sum\limits_{p\in H'}a_{pi}}{\sum\limits_{l\in N}a_{li}}(c_{\beta\alpha}-c_{\omega\alpha})>0, \omega\in[1:\beta-1]\setminus\{\xi\}  \\
    \sum\limits_{\alpha=1}^{k}\sum\limits_{H'\in S_\theta^\alpha}\frac{\sum\limits_{p\in H'}a_{pi}}{\sum\limits_{l\in N}a_{li}}(c_{\beta\alpha}-c_{\xi+1\alpha})\geq0, \omega\in[\beta+1:k]\\
c_{\xi h_\theta}-c_{\beta h_\theta}\geq\sum\limits_{\alpha\neq h_\theta}\frac{\sum\limits_{H'\in S_\theta^\alpha}\sum\limits_{p\in H'}a_{pj}}{\sum\limits_{H'\in S_\theta^{h_\theta}}\sum\limits_{p\in H'}a_{pj}}(c_{\xi\alpha}-c_{\beta\alpha})\\
c_{\xi h_\theta}-c_{\beta h_\theta}<\sum\limits_{\alpha\neq h_\theta}\frac{\sum\limits_{H'\in S_\theta^\alpha}\sum\limits_{p\in H'}a_{pi}}{\sum\limits_{H'\in S_\theta^{h_\theta}}\sum\limits_{p\in H'}a_{pi}}(c_{\xi\alpha}-c_{\beta\alpha}).
  \end{array}
\right.
\end{equation}

It can be seen that the last two lines of \eqref{4.3} and \eqref{4.4} are mutually exclusive with the last two lines of \eqref{2.3} and \eqref{4.5} respectively. Thus, we can find a suitable payoff matrix $C$ to make \eqref{2.3} and \eqref{4.5} or \eqref{4.3} and \eqref{4.4} hold as long as $\sum\limits_{\alpha\neq h_\theta}\frac{\sum\limits_{H'\in S_\theta^\alpha}\sum\limits_{p\in H'}a_{pj}}{\sum\limits_{H'\in S_\theta^{h_\theta}}\sum\limits_{p\in H'}a_{pj}}(c_{\xi\alpha}-c_{\beta\alpha})\neq\sum\limits_{\alpha\neq h_\theta}\frac{\sum\limits_{H'\in S_\theta^\alpha}\sum\limits_{p\in H'}a_{pi}}{\sum\limits_{H'\in S_\theta^{h_\theta}}\sum\limits_{p\in H'}a_{pi}}(c_{\xi\alpha}-c_{\beta\alpha})$.
When $\sum\limits_{\alpha\neq h_\theta}(\frac{\sum\limits_{H'\in S_\theta^\alpha}\sum\limits_{p\in H'}a_{pj}}{\sum\limits_{H'\in S_\theta^{h_\theta}}\sum\limits_{p\in H'}a_{pj}}-\frac{\sum\limits_{H'\in S_\theta^\alpha}\sum\limits_{p\in H'}a_{pi}}{\sum\limits_{H'\in S_\theta^{h_\theta}}\sum\limits_{p\in H'}a_{pi}})(c_{\xi\alpha}-c_{\beta\alpha})>0$, if the payoff matrix $C$ satisfies $$\left\{
                  \begin{array}{ll}
                    c_{\xi\alpha}-c_{\omega\alpha}>0,\;\alpha\in[1:k],\omega\in[1:\xi-1]\setminus\{\beta\} \\
c_{\xi\alpha}-c_{\omega\alpha}\geq0,\;\alpha\in[1:k],\omega\in[\xi+1,k] \\
                   c_{\xi h_\theta}-c_{\beta h_\theta}>\sum\limits_{\alpha\neq h_\theta}\frac{\sum\limits_{H'\in S_\theta^\alpha}\sum\limits_{p\in H'}a_{pi}}{\sum\limits_{H'\in S_\theta^{h_\theta}}\sum\limits_{p\in H'}a_{pi}}(c_{\xi\alpha}-c_{\beta\alpha})\\
c_{\xi h_\theta}-c_{\beta h_\theta}\leq\sum\limits_{\alpha\neq h_\theta}\frac{\sum\limits_{H'\in S_\theta^\alpha}\sum\limits_{p\in H'}a_{pj}}{\sum\limits_{H'\in S_\theta^{h_\theta}}\sum\limits_{p\in H'}a_{pj}}(c_{\xi\alpha}-c_{\beta\alpha}),
                  \end{array}
                \right.
$$
or
$$\left\{
                  \begin{array}{ll}
                    c_{\xi\alpha}-c_{\omega\alpha}>0,\;\alpha\in[1:k],\omega\in[1:\xi-1] \\
c_{\xi\alpha}-c_{\omega\alpha}\geq0,\;\alpha\in[1:k],\omega\in[\xi+1,k]\setminus\{\beta\}  \\
                   c_{\xi h_\theta}-c_{\beta h_\theta}\geq\sum\limits_{\alpha\neq h_\theta}\frac{\sum\limits_{H'\in S_\theta^\alpha}\sum\limits_{p\in H'}a_{pi}}{\sum\limits_{H'\in S_\theta^{h_\theta}}\sum\limits_{p\in H'}a_{pi}}(c_{\xi\alpha}-c_{\beta\alpha})\\
c_{\xi h_\theta}-c_{\beta h_\theta}<\sum\limits_{\alpha\neq h_\theta}\frac{\sum\limits_{H'\in S_\theta^\alpha}\sum\limits_{p\in H'}a_{pj}}{\sum\limits_{H'\in S_\theta^{h_\theta}}\sum\limits_{p\in H'}a_{pj}}(c_{\xi\alpha}-c_{\beta\alpha}),
                  \end{array}
                \right.
$$
then a scenario will occur where player $i$ chooses strategy $x_i^\xi$ while player $j$ does not choose strategy $x_i^\xi$ at the next time.

When $\sum\limits_{\alpha\neq h_\theta}(\frac{\sum\limits_{H'\in S_\theta^\alpha}\sum\limits_{p\in H'}a_{pj}}{\sum\limits_{H'\in S_\theta^{h_\theta}}\sum\limits_{p\in H'}a_{pj}}-\frac{\sum\limits_{H'\in S_\theta^\alpha}\sum\limits_{p\in H'}a_{pi}}{\sum\limits_{H'\in S_\theta^{h_\theta}}\sum\limits_{p\in H'}a_{pi}})(c_{\xi\alpha}-c_{\beta\alpha})<0$, if the payoff matrix $C$ meets
$$\left\{
                  \begin{array}{ll}
                    c_{\beta\alpha}-c_{\omega\alpha}>0,\;\alpha\in[1:n],\omega\in[1:\beta-1] \\
c_{\beta\alpha}-c_{\omega\alpha}\geq0,\;\alpha\in[1:n],\omega\in[\beta+1:k]\setminus\{\xi\} \\
                   c_{\xi h_\theta}-c_{\beta h_\theta}>\sum\limits_{\alpha\neq h_\theta}\frac{\sum\limits_{H'\in S_\theta^\alpha}\sum\limits_{p\in H'}a_{pj}}{\sum\limits_{H'\in S_\theta^{h_\theta}}\sum\limits_{p\in H'}a_{pj}}(c_{\xi\alpha}-c_{\beta\alpha})\\
c_{\xi h_\theta}-c_{\beta h_\theta}\leq\sum\limits_{\alpha\neq h_\theta}\frac{\sum\limits_{H'\in S_\theta^\alpha}\sum\limits_{p\in H'}a_{pi}}{\sum\limits_{H'\in S_\theta^{h_\theta}}\sum\limits_{p\in H'}a_{pi}}(c_{\xi\alpha}-c_{\beta\alpha}),
                  \end{array}
                \right.
$$
or
$$\left\{
                  \begin{array}{ll}
                    c_{\beta\alpha}-c_{\omega\alpha}>0,\;\alpha\in[1:n],\omega\in[1:\beta-1]\setminus\{\xi\} \\
c_{\beta\alpha}-c_{\omega\alpha}\geq0,\;\alpha\in[1:n],\omega\in[\beta+1:k]\\
                   c_{\xi h_\theta}-c_{\beta h_\theta}\geq\sum\limits_{\alpha\neq h_\theta}\frac{\sum\limits_{H'\in S_\theta^\alpha}\sum\limits_{p\in H'}a_{pj}}{\sum\limits_{H'\in S_\theta^{h_\theta}}\sum\limits_{p\in H'}a_{pj}}(c_{\xi\alpha}-c_{\beta\alpha})\\
c_{\xi h_\theta}-c_{\beta h_\theta}<\sum\limits_{\alpha\neq h_\theta}\frac{\sum\limits_{H'\in S_\theta^\alpha}\sum\limits_{p\in H'}a_{pi}}{\sum\limits_{H'\in S_\theta^{h_\theta}}\sum\limits_{p\in H'}a_{pi}}(c_{\xi\alpha}-c_{\beta\alpha}),
                  \end{array}
                \right.
$$
then player $i$ will choose strategy $x_i^\beta$ at the next time, but player $j$ will not.



Therefore, in order to prevent the situation where players $i$ and $j$ choose different strategies in the next time under the profile $s^{i_\theta}$, the only way is to make \begin{equation}\label{3.8}\sum\limits_{\alpha\neq h_\theta}(\frac{\sum\limits_{H'\in S_\theta^\alpha}\sum\limits_{p\in H'}a_{pj}}{\sum\limits_{H'\in S_\theta^{h_\theta}}\sum\limits_{p\in H'}a_{pj}}-\frac{\sum\limits_{H'\in S_\theta^\alpha}\sum\limits_{p\in H'}a_{pi}}{\sum\limits_{H'\in S_\theta^{h_\theta}}\sum\limits_{p\in H'}a_{pi}})(c_{\xi\alpha}-c_{\beta\alpha})=0.\end{equation}

There are $k^2-k$ possible combinations of $x_i^\xi, \xi\in[1:k]$ and $x_i^\beta, \beta\in[1:n]\setminus\{\xi\}$, from which $\frac{k^2-k}{2}$ equations similar to \eqref{3.8} can be obtained. That is, for any $\omega\in[1:k-1]$,
\begin{equation}\label{2.5}
    \sum\limits_{\alpha\neq h_\theta}(\frac{\sum\limits_{H'\in S_\theta^\alpha}\sum\limits_{p\in H'}a_{pj}}{\sum\limits_{H'\in S_\theta^{h_\theta}}\sum\limits_{p\in H'}a_{pj}}-\frac{\sum\limits_{H'\in S_\theta^\alpha}\sum\limits_{p\in H'}a_{pi}}{\sum\limits_{H'\in S_\theta^{h_\theta}}\sum\limits_{p\in H'}a_{pi}})(c_{\omega+1\alpha}-c_{\omega\alpha})=0.  \\
\end{equation}
It is easy to obtain that \eqref{2.5} holds for all payoff matrices if and only if $$\frac{\sum\limits_{H'\in S_\theta^\alpha}\sum\limits_{p\in H'}a_{pj}}{\sum\limits_{H'\in S_\theta^{h_\theta}}\sum\limits_{p\in H'}a_{pj}}=\frac{\sum\limits_{H'\in S_\theta^\alpha}\sum\limits_{p\in H'}a_{pi}}{\sum\limits_{H'\in S_\theta^{h_\theta}}\sum\limits_{p\in H'}a_{pi}}, \forall\alpha\in[2:k].$$
Then there must exist a real number $q_\theta$ such that
\begin{equation}\label{5.1}\sum\limits_{H'\in S_\theta^\alpha}\sum\limits_{p\in H'}a_{pj}=q_\theta(\sum\limits_{H'\in S_\theta^\alpha}\sum\limits_{p\in H'}a_{pi}), \forall \alpha\in[1:k].\end{equation}

For any profile in the set $S'$, we derive an equation similar to \eqref{5.1}. Combining these $k^M$ equations together, after a simple calculation, we can get
$$\sum\limits_{p\in H_\kappa}a_{pj}=q_\theta\sum\limits_{p\in H_\kappa}a_{pi}, \forall \kappa\in[1:M].$$
Thus, for arbitrary $H,H'\in\mathcal{H}, i,j\in H$ and $\kappa\in[1:M]$, we have $$\frac{q_\theta\sum\limits_{p\in H_\kappa}a_{pi}}{q_\theta\sum\limits_{\kappa=1}^{M}\sum\limits_{p\in H_\kappa}a_{pi}}=\frac{\sum\limits_{p\in H'}a_{pi}}{\sum\limits_{l\in N}a_{li}}=\frac{\sum\limits_{p\in H'}a_{pj}}{\sum\limits_{l\in N}a_{lj}}.$$
In conclusion, the necessity is proved.$\hfill\square$

The proof of Theorem \ref{6.2}: (Sufficiency) From the construction $\hat{L}$, it follows that if $[\hat{L}]_{\rho\epsilon}>0$, there must be at least one control $\delta_{k^m}^{i}$ such that profile $\delta_{k^M}^\epsilon$ can reach profile $\delta_{k^M}^\rho$. Since $\Xi^T\mathrm{Col}_\epsilon(\hat{L}^{t_1})
=[\hat{L}^{t_1}]_{\lambda_1\epsilon}+[\hat{L}^{t_1}]_{\lambda_2\epsilon}+
\cdots+[\hat{L}^{t_1}]_{\lambda_q\epsilon}$, if $\Xi^T\mathrm{Col}_\epsilon(\hat{L}^{t_1})>0, \forall \;\epsilon\in[1:k^M]$ holds, profile $\delta_{k^M}^\epsilon\in\Delta_{k^M}$ can reach $I_c(\bar{\Theta})$ under some control $\delta_{k^{m}}^{i}$. Then there must exist $q$ non-negative real numbers $b_1^\epsilon,b_2^\epsilon,\cdots,b_{q}^\epsilon$ such that $\mathrm{Col}_\epsilon(\hat{L}^{t_1})=b_1^\epsilon\delta_{k^M}^{\lambda_1}+b_2^\epsilon\delta_{k^M}^{\lambda_2}
+\cdots+b_{q}^\epsilon\delta_{k^M}^{\lambda_{q}}$. Since $I_c(\bar{\Theta})$ is a control invariant set, one has
\begin{align*}
&\Xi^T\mathrm{Col}_\epsilon(\hat{L}^{t_1+1})\\
&=\Xi^T\cdot\hat{L}\cdot\mathrm{Col}_\epsilon(\hat{L}^{t_1})\\
&=\Xi^T(b_1^\epsilon\hat{L}\delta_{k^M}^{\lambda_1}+b_2^\epsilon\hat{L}\delta_{k^{M}}^{\lambda_2}
+\cdots+b_{q}^\epsilon\hat{L}\delta_{k^M}^{\lambda_q})\\
&=\Xi^T(b_1^\epsilon\mathrm{Col}_{\lambda_1}(\hat{L})+b_2^\epsilon\mathrm{Col}_{\lambda_2}(\hat{L})+\cdots+
b_{q}^\epsilon\mathrm{Col}_{\lambda_q}(\hat{L}))\\
&=b_1^\epsilon\Xi^T\mathrm{Col}_{\lambda_1}(\hat{L})+b_2^\epsilon\Xi^T\mathrm{Col}_{\lambda_2}(\hat{L})+\cdots\\&\quad+
b_{q}^\epsilon\Xi^T\mathrm{Col}_{\lambda_q}(\hat{L})>0
\end{align*}
Obviously, $\Xi^T\mathrm{Col}_\epsilon(\hat{L}^{t_1+t})>0$ for all positive integer $t$, which shows that for any profile $\delta_{k^M}^\epsilon\in\Delta_{k^{M}}$ after $t_1+t$ steps still evolves to $I_c(\bar{\Theta})$ under some control sequence. So system \eqref{3.5} converges to $I_c(\bar{\Theta})$. That is, the WNEG with algebraic form \eqref{3.5} reaches strategy consensus.

(Necessity) We prove it by contradiction. If $\hat{\Theta}=\emptyset$, it means that all super players do not have the same strategy.
Assuming $I_c(\bar{\Theta})=\emptyset$, there must not exist an integer $t_1$ such that $y(t_1)\in \bar{\Theta}$. These contradict the fact that the WNEG with algebraic
form \eqref{3.5} achieves strategy consensus.
Suppose the third condition of the Theorem \ref{6.2} does not hold, there must exist a profile $\delta_{k^M}^{\xi}$ that cannot $(\alpha-q)$-step reach $I_c(\bar{\Theta})$ under any control. This contradicts the fact that the system \eqref{3.5} converges to the set $I_c(\bar{\Theta})$.$\hfill\square$

\subsection{Algorithm}
\begin{algorithm}
         \caption{Verify the strategy consensus of the WNEG with algebraic form \eqref{3.5} }\label{7.6}
         \begin{algorithmic}[1]
                    \Require $\bar{\Theta}=\{\delta_{k^M}^{\theta_1},\cdots,\delta_{k^M}^{\theta_a}\}, \hat{S}=\{\delta_{k^M}^{\beta_1},\cdots,\delta_{k^M}^{\beta_\alpha}\},\bar{L}$.
\State Split $\bar{L}$ into $k^m$ equal blocks as $\bar{L}=[\bar{L}_1,\bar{L}_2,\cdots,\bar{L}_{k^m}]$
and Compute $\hat{L}=\sum\limits_{\tau=1}^{k^m}\bar{L}_\tau$
\State  $\mathbf{Initialization:}$ $t\leftarrow 0$, $\mathcal{R}_t(\bar{\Theta})=\bar{\Theta} $

\While{$\mathcal{R}_t(\bar{\Theta})\neq \varnothing$}
\State  Compute $\mathcal{R}_{t+1}(\delta_{k^{M}}^{\omega}):=\{\delta_{k^{M}}^\rho|[\hat{L}]_{\rho\omega}>0, \delta_{k^{M}}^{\omega},\delta_{k^{M}}^{\rho}\in\mathcal{R}_t(\bar{\Theta}) \}$ and $\mathcal{R}_{t+1}(\bar{\Theta})=\bigcup\limits_{\delta_{k^{M}}^{\omega}\in\mathcal{R}_t(\bar{\Theta})}
\mathcal{R}_{t+1}(\delta_{k^M}^{\omega})$
\While {$\mathcal{R}_{t+1}(\bar{\Theta})\neq\emptyset\subseteq \mathcal{R}_{t}(\bar{\Theta})$}
\State $t+1\leftarrow t$
\EndWhile
\EndWhile
\State Let $I_c(\bar{\Theta})=\mathcal{R}_{t_0}(\bar{\Theta})=\{\delta_{k^M}^{\lambda_1},\delta_{k^M}^{\lambda_2},\cdots
\delta_{k^M}^{\lambda_q}\}$ and calculate $\Xi=\sum\limits_{i=1}^q\delta_{k^M}^{\lambda_i}$
\State  Compute $\hat{L}^{\alpha-q}$ and $\Xi^T\mathrm{Col}_{\omega}(\hat{L}^{\alpha-q}), \delta_{k^{M}}^\omega\in \hat{S}\setminus\mathcal{R}_{t_0}(\bar{\Theta})$
\If {$\Xi^T\mathrm{Col}_{\omega}(\hat{L}^{\alpha-q})>0, \forall \delta_{k^{M}}^\omega\in \hat{S}\setminus\mathcal{R}_{t_0}(\bar{\Theta})$} the WNEG with algebraic form \eqref{3.5} achieves strategy consensus.
\State \textbf{otherwise}, the WNEG cannot achieve strategy consensus.
                   \State  \textbf{break}
                   \EndIf
 \end{algorithmic}
\end{algorithm}

\begin{algorithm}
         \caption{Verify the strategy optimization of the WNEG with algebraic form \eqref{3.5}}\label{suanfa1}
         \begin{algorithmic}[1]
                    \Require $I_c(\Omega)$, $\bar{L}$; $T_{\overline{W}_0|\overline{W}_0}$.
                   \Ensure  $\hat{T}_1$.
 \State  $\mathbf{Initialization:}$ $t\leftarrow 0$, $\overline{W}_t \leftarrow I_{c}(\Omega)$
                   \While{$\overline{W}_t\neq \varnothing$}
\State  $W_{t+1}=\Delta_{k^M}\setminus\bigcup\limits_{i=0}^{t}\overline{W}_{i}$
                   \State  Compute the truth matrix $T_{\overline{W}_{t}|W_{t+1}}$ as
                   $$[T_{\overline{W}_{t}|W_{t+1}}]_{ij}=
\begin{cases}
1,& \mathrm{if}\;\bar{L}\delta_{k^{m}}^{i}\delta_{k^M}^{j}\in \overline{W}_{t},\;\forall\delta_{k^M}^{j}\in W_{t+1},\\
0,& \mathrm{otherwise}.
\end{cases}$$
                   \State Construct $\overline{W}_{t+1}=\{\delta_{k^M}^j|\mathrm{Col}_j(T_{\overline{W}_{t}|W_{t+1}})\neq \mathbf{0}_{k^m}\}$
                   \State $t\leftarrow t+1$
                   \EndWhile
                   \State  \ \ \ $t_2\leftarrow t$
                   \State  \ \ \ Compute $\widehat{T}_1=T_{\overline{W}_0|\overline{W}_0}+\sum_{\mu=1}^{t_2}T_{\overline{W}_{\mu-1}|\overline{W}_{\mu}}$
                   \State  \textbf{return} $\widehat{T}_1$
\If {$\mathrm{Col}_j(\hat{T}_1)\neq\mathbf{0}_{k^m}, \forall\;j\in[1:k^M]$,} the WNEG with algebraic form \eqref{3.5} achieves strategy optimization objective.
\State \textbf{otherwise}, the WNEG cannot achieve the strategy optimization objective.
                   \State  \textbf{break}
                   \EndIf
         \end{algorithmic}
\end{algorithm}

\begin{algorithm}
         \caption{Verify the controllability of the WNEG with algebraic form \eqref{3.5}}\label{7.5}
         \begin{algorithmic}[1]
                    \Require $\bar{L}$, $\delta_{k^M}^i$, $\delta_{k^M}^j$, $K$.
                   \Ensure  $\hat{L}^K, \sum\limits_{\mu=1}^{k^M}\hat{L}^{\mu}$.
 \State  Split $\bar{L}$ into $k^m$ equal blocks as $\bar{L}=[\bar{L}_1,\bar{L}_2,\cdots,\bar{L}_{k^m}]$
\State Compute $\hat{L}=\sum\limits_{\tau=1}^{k^m}\bar{L}_\tau$
and $\hat{L}^K$
\If{ $[\hat{L}^K]_{ij}>0$,} go to Step 6.
\State \textbf{otherwise}
                 \State  \textbf{break}
\State Compute $\sum\limits_{\mu=1}^{k^M}\hat{L}^\mu$
\If {$\sum\limits_{\mu=1}^{k^M}\mathrm{Col}_j(\hat{L}^\mu)>0$,} go to Step 10.
                   \State \textbf{otherwise}
                   \State  \textbf{break}
\If {$\sum\limits_{\mu=1}^{k^M}\hat{L}^\mu>0$,} the WNEG with algebraic form \eqref{3.5} is controllable.
 \State \textbf{otherwise} the WNEG is not controllable.
\State  \textbf{return} $\hat{L}^K, \sum\limits_{\mu=1}^{k^M}\hat{L}^\mu$
 \State  \textbf{break}
\EndIf
\EndIf
\EndIf
         \end{algorithmic}
\end{algorithm}
\subsection{Profile feedback control}
The profile feedback control for strategy consensus is designed as follows.
Define $\Psi_\tau(I_c(\bar{\Theta}))$ is the set of profiles that can reach $I_c(\bar{\Theta})$ at $\tau$-step and $\Psi_0(I_c(\bar{\Theta}))=I_c(\bar{\Theta})$.
We compute $\Psi_\tau(I_c(\bar{\Theta})):=\{\delta_{k^M}^\epsilon|\Xi^T \mathrm{Col}_{\epsilon}(\hat{L})>0, \delta_{k^M}^\epsilon\in\Delta_{k^M}\setminus\bigcup\limits_{i=0}^{\tau-1}\Psi_i(I_c(\bar{\Theta})) \}$. Then, the  profile feedback gain matrix is $\bar{G}=\delta_{k^{m}}[\beta_1,\beta_2,\cdots,\beta_{k^M}]$, where $\beta_j\in\{\eta|\mathrm{Col}_\epsilon(\bar{L}_\eta)\in\Psi_{\tau-1}(I_c(\bar{\Theta}))
, \delta_{k^{M}}^\epsilon\in\Psi_\tau(I_c(\bar{\Theta}))\}, j, \epsilon\in[1:k^M],\tau\in[1:t_1]$. Then we give the profile feedback control that enables the WNEG with algebraic form \eqref{10.3} to achieve strategy consensus. For any profile $y(t)=\ltimes_{j=1}^My_j(t)=\delta_{k^M}^j, j\in[1:M]$, we compute the profile $x(t)=\ltimes_{i=1}^nx_i(t)=\delta_{k^n}^{\pi_j}$, where  $x_i(t)=x_{j_\iota}(t)=\Lambda_M^jy(t), \forall i\in H_j, \iota\in[1:n_j], j\in[1:M]$ and
$\Lambda_M^j=\left\{
                                          \begin{array}{ll}
                                            I_{k}\otimes\mathbf{1}_{k^{M-1}},& j=1\\
\mathbf{1}_{k^{j-1}}\otimes I_{k}\otimes \mathbf{1}_{k^{M-j}}, &  j\in[2:M]\\
                                            \mathbf{1}_{k^{M-1}}\otimes I_k, & j=M
                                          \end{array}
                                        \right.
$.
The profile feedback gain matrix is $$\left\{
                                      \begin{array}{ll}
                                        \mathrm{Col}_{\pi_j}(G)\in\{\mathbf{1}_{k^{m'}}\ltimes_{\alpha=1}^{m'}u_\alpha, |u_{\alpha}=\hat{u}_q=\Lambda_m^q\beta_j, \\ \quad \quad \quad \; \forall i_\alpha\in H_{j_q}, \alpha\in[1:m'], q\in[1:m]\}, \\
                                        \mathrm{Col}_\pi(G)\in\{\mathbf{1}_{k^{m'}}\ltimes_{\alpha=1}^{m'}u_\alpha|u_a=u_b, \forall i_a, i_b\in H_{v_q},\\ \;\quad \quad \quad a,b\in[1:m'], q\in[1:m],u_\alpha\in\Delta_k\}, \pi\neq\pi_j.
                                      \end{array}
                                    \right.$$


The profile feedback control for strategy optimization is designed as follows.
Then, the  profile feedback gain matrix is $\bar{G}=\delta_{k^{m}}[\beta_1,\beta_2,\cdots,\beta_{k^M}]$, where
$$\left\{
        \begin{array}{ll}
          \beta_j\in\{i|[T_{\overline{W}_0|\overline{W}_0}]_{ij}=1\}, & \delta_{k^M}^j\in I_c(\Omega)\\
          \beta_j\in \{i|[T_{\overline{W}_{t-1}|W_{t}}]_{ij}=1\}, & \delta_{k^M}^j\in W_t, t\in[1:t_2].
        \end{array}
      \right.
$$
The methods for designing the profile feedback control of system \eqref{10.3} based on that of system \eqref{3.5} are the same as those in the previous section, and will not be repeated here.


\end{document}